\documentclass[cleveref,thm-restate, a4paper]{lipics-v2021}

\usepackage{graphicx}
\usepackage{xcolor}
\usepackage{cleveref}
\usepackage{xspace}
\usepackage{bbm}
\usepackage{mathtools}
\usepackage[style=english]{csquotes}

\usepackage{paralist}

\nolinenumbers
\pdfoutput=1 
\hideLIPIcs  

\newcommand{\Oh}{\ensuremath{\mathcal{O}}}
\newcommand{\RR}{\ensuremath{\mathbb{R}}}
\newcommand{\F}{{Fréchet distance }}
\newcommand{\Fm}{{Fréchet distance}}
\newcommand{\cpP}{{\ensuremath{p}}} 
\newcommand{\cpQ}{{\ensuremath{q}}} 
\newcommand{\X}{\ensuremath{X}}
\newcommand{\Y}{\ensuremath{Y}}
\newcommand{\pre}[1]{\ensuremath{\mathrm{pre}(#1)}}
\newcommand{\suf}[1]{\ensuremath{\mathrm{suf}(#1)}}
\newcommand{\vpre}{\ensuremath{v_\mathrm{pre}}}
\newcommand{\vsuf}{\ensuremath{v_\mathrm{suf}}}
\newcommand{\wpre}{\ensuremath{w_\mathrm{pre}}}
\newcommand{\wsuf}{\ensuremath{w_\mathrm{suf}}}
\newcommand{\ePQ}[2]{{minimal \ensuremath{#1}-matcher on \ensuremath{#2}\xspace}}

\usepackage{soul}

\definecolor{darkcyan}{rgb}{0.0, 0.55, 0.55}

\DeclareMathOperator*{\im}{im}

\title{Transforming Dogs on the Line: On the Fréchet Distance Under Translation or Scaling in 1D}
 \titlerunning{Transforming Dogs on the Line}
\author{Lotte Blank}{University of Bonn, Germany}{lblank@uni-bonn.de}{https://orcid.org/0000-0002-6410-8323}{Funded by the Deutsche Forschungsgemeinschaft (DFG, German Research Foundation) – 459420781 (FOR AlgoForGe)}

\author{Jacobus Conradi}{University of Bonn, Germany}{conradi@cs.uni-bonn.de}{https://orcid.org/0000-0002-8259-1187}{Funded by the iBehave Network: Sponsored by the Ministry of Culture and Science of the State of North Rhine-Westphalia.}
\author{Anne Driemel}{University of Bonn, Germany}{driemel@cs.uni-bonn.de}{https://orcid.org/0000-0002-1943-2589}{Affiliated with Lamarr Institute for Machine Learning and Artificial Intelligence.}
\author{Benedikt Kolbe}{University of Bonn, Germany}{bkolbe@uni-bonn.de}{https://orcid.org/0009-0005-0440-4912}{This work was partially supported by the Lamarr Institute for Machine Learning and Artificial Intelligence.}
\author{André Nusser}{Université Côte d'Azur, CNRS, Inria, France}{andre.nusser@cnrs.fr}{https://orcid.org/0000-0002-6349-869X}{This work was supported by the French government through the France 2030 investment plan managed by the National Research Agency (ANR), as part of the Initiative of Excellence of Université Côte d'Azur under reference number ANR-15-IDEX-01.}
\author{Marena Richter}{University of Bonn, Germany}{marenarichter@uni-bonn.de}{https://orcid.org/0009-0007-8250-266X}{}

\authorrunning{Blank, Conradi, Driemel, Kolbe, Nusser and Richter}

\Copyright{Lotte Blank, Jacobus Conradi, Anne Driemel, Benedikt Kolbe, André Nusser, Marena Richter} 

\begin{document}

\maketitle

\begin{abstract}

The Fréchet distance is a computational mainstay for comparing polygonal curves.
The Fréchet distance under translation, which is a translation invariant version, considers the similarity of two curves independent of their location in space.
It is defined as the minimum Fréchet distance that arises from allowing arbitrary translations of the input curves.
This problem and numerous variants of the Fréchet distance under some transformations have been studied, with more work concentrating on the discrete Fréchet distance, leaving a significant gap between the discrete and continuous versions of the Fréchet distance under transformations.
Our contribution is twofold:
First, we present an algorithm for the Fréchet distance under translation on 1-dimensional curves of complexity $n$ with a running time of $\mathcal{O}(n^{8/3} \log^3 n)$.
To achieve this, we develop a novel framework for the problem for 1-dimensional curves, which also applies to other scenarios and leads to our second contribution.
We present an algorithm with the same running time of $\mathcal{O}(n^{8/3} \log^3 n)$ for the Fréchet distance under scaling for 1-dimensional curves.
For both algorithms we match the running times of the discrete case and improve the previously best known bounds of $\tilde{\mathcal{O}}(n^4)$.
Our algorithms rely on technical insights but are conceptually simple, essentially reducing the continuous problem to the discrete case across different length scales.

\end{abstract}

\section{Introduction}

The Fréchet distance \cite{frechet1906} is one of the most well-studied distance measures for polygonal curves, with a long history in both applications \cite{BrakatsoulasPSW05,Cleasby2019,DBLP:conf/gis/BuchinDLN19} and algorithmic research \cite{AltG95,DBLP:journals/dcg/DriemelHW12,DBLP:journals/talg/Har-PeledR14}.
For many areas of interest, including the analysis of stock market trends, electrocardiograms or electroencephalograms, the underlying behavior to be analyzed crucially involves the behavior of 1-dimensional curves, or \emph{time series} \cite{BIRAN2023104575,NIU2013489}.
One of the advantages that the Fréchet distance offers over competing distance measures such as the Hausdorff distance is that the Fréchet distance considers a joint traversal of both curves, which naturally takes the orientation and ordering of the vertices of both curves into account.
For 1-dimensional curves, the difference is particularly jarring, as representing a \emph{continuous} curve by the points lying on it essentially discards all information concerning the underlying dynamics when measuring the Hausdorff distance. 
Intuitively, the Fréchet distance between two curves can be thought of the shortest possible length of a rope that can connect two climbers that cooperatively travel along their respective routes from start to end.
Another well-known analogy stems from the metaphor of a person walking a dog. In this setting, the connecting rope corresponds to a dog leash and we ask for the shortest leash that can be used for the walk.

The algorithmic complexity of the computation of the Fréchet distance has been thoroughly investigated and is well understood. For curves of complexity $n$, the Fréchet distance can be computed in time $\tilde{\mathcal{O}}(n^2)$ \cite{AltG95,BuchinBMM14,DBLP:journals/corr/abs-2407-05231}, where $\tilde{\mathcal{O}}$ hides polylogarithmic factors in $n$. Furthermore, this is hard in the sense that there cannot exist an $\mathcal{O}(n^{2-\epsilon})$ time algorithm for any $\epsilon>0$, unless the Strong Exponential Time Hypothesis (SETH) fails~\cite{Bringmann14,DBLP:conf/soda/BuchinOS19}. 

The Fréchet distance has been studied in many different contexts, leading to a plethora of variants, inspired both by applications as well as theoretical considerations, see~\cite{DBLP:conf/compgeom/BuchinNW22,Conradi2024, DBLP:journals/siamcomp/DriemelH13, DBLP:journals/jocg/FiltserK20}.
Each of the available variants fits into one of two classes, according to whether only the distances of vertices of the curves are considered (the discrete version), or whether the computed distance depends on all points along the straight line segments connecting them (the continuous variant).

While the continuous Fréchet distance naturally considers all points between two vertices and tends to give better practical results (e.g., no discretization artifacts), the discrete Fréchet distance is often easier to handle algorithmically as well as implementation-wise.
One overarching goal of Fréchet related research is to match the algorithmic results that are attained in the discrete version with those in the continuous case.

\subparagraph*{Fréchet distance under translation or scaling.}

An important aspect of detecting movement patterns in different application  areas is to consider the movement independent of its absolute location and scale, i.e., we want to know for which location and scale do the patterns look most similar.
Consider, for example, a cardiogram of a patient with heart arrhythmia~\cite{BIRAN2023104575}.
Intuitively, some of the shared characteristics of the observed heartbeats do not depend on their absolute location/scaling but instead on the relative movement of the wave.
Similarly, identifying a trend in the stock market does not necessarily depend on the exact values of the stocks involved.
Instead, a group of stocks may display similar relative behavior although their absolute values are far apart, in which case a natural approach is to allow arbitrary translation and rescaling of the curves to distill their concurrent behavior.
These motivations directly lead to the idea of the Fréchet distance under translation or scaling.
Concretely, the Fréchet distance under translation (resp. scaling) is the minimum Fréchet distance that we obtain by allowing an arbitrary translation (resp. scaling) of one curve.

%
%

\subparagraph*{Related work.}


As mentioned previously, due to the simpler setting, there are often more results in the discrete setting, as is also the case for the \emph{discrete} Fréchet distance under translation.
This problem was previously mainly considered in the Euclidean plane.
The first algorithm due to Jiang, Xu, and Zhu~\cite{JiangXZ08} builds an arrangement of size $\Oh(n^4)$ over the translation space, and then queries an arbitrary representative of each cell of this arrangement with the quadratic-time Fréchet distance algorithm, leading to an $\Oh(n^6 \log n)$ running time (where the logarithmic factors stems from the usage of parametric search).
Avraham, Kaplan, and Sharir~\cite{DBLP:journals/corr/AvrahamKS15} crucially observed that instead of re-computing the Fréchet distance for each cell of the arrangement, one can instead formulate the problem as an online dynamic reachability problem on a grid graph and provide a data structure with update time $\Oh(n)$, leading to an $\Oh(n^5 \log n)$ algorithm.
The current best-known algorithm observes that these updates are offline, i.e., they can be known in advance, and design a tailored data structure for this problem reducing the update time to $\Oh(n^{2/3} \log^2 n)$, leading to a running time of $\Oh(n^{4+\frac{2}{3}} \log^3 n)$.
\footnote{Meanwhile, this problem was also tackled from the algorithm engineering side, making use of Lipschitz optimization instead of dynamic graph algorithms \cite{BKN20}.}
This result is complemented by a lower bound of $\mathcal{O}(n^{4-o(1)})$, conditional on the Strong Exponential Time Hypothesis (SETH), for the discrete Fréchet distance under translation in the plane~\cite{DBLP:journals/talg/BringmannKN21}.

The currently best algorithm for discrete time series that is explicitly described in the literature has running time $\mathcal{O}(n^3)$~\cite{Filtser2020}, based on the solution to the decision version of the problem in~\cite{DBLP:journals/corr/AvrahamKS15}.
We note that the approach of Bringmann, Künnemann, and Nusser~\cite{DBLP:journals/talg/BringmannKN21} is also directly applicable to the 1-dimensional problem, which leads to an $\Oh(n^{{8}/{3}} \log^3 n)$ algorithm.\footnote{This is a result of the arrangement size in $\mathbb{R}$ being $n^2$, while in $\mathbb{R}^2$ it is $n^4$. Hence, the running time is reduced by a factor of $n^2$.}


The \emph{continuous} Fréchet distance under translations was first studied in 2001~\cite{DBLP:conf/stacs/AltKW01,AltKW04,efrat2001pattern}.
Shortly after, Wenk~\cite{wenk2002phd} considered the Fréchet distance under scaling and developed the currently best published algorithms for both variants (translation or scaling) in the continuous setting in $d$-dimensional Euclidean space.
For time series, these algorithms achieve a running time of $\mathcal{O}(n^5\log n)$.
We note that for time series there is an unpublished $\tilde{\Oh}(n^4)$ algorithm for both problems that can be considered folklore.
We briefly sketch the idea here.
First, we build the arrangement of intervals separated by events at which the decision problem for the Fréchet distance under translation/scaling is subject to change.
The crux is that in one dimension, there are only $\Oh(n^2)$ such events, as the edge-vertex-vertex events that appear in higher dimensions can be omitted as they are captured by the vertex-vertex events (see \cref{cor:CompuTranslationArrangement} and \cref{cor:CompuScalingArrangement} for more details).
Thus, computing the Fréchet distance for an arbitrary representative translation/scaling in each cell together with parametric search \cite{Cole87} yields an $\tilde{\Oh}(n^4)$ algorithm.
We want to stress that, while indeed the events in one dimension only depend on vertex-vertex alignments, this does not directly enable the usage of the grid reachability data structures of~\cite{DBLP:journals/corr/AvrahamKS15}~and~\cite{DBLP:journals/talg/BringmannKN21}, as we still have to compute the \emph{continuous} Fréchet distance for each cell.

 
\subsection{Our Contributions}

We present the first new result on the continuous Fréchet distance under translation and the continuous Fréchet distance under scaling since the results by Wenk~\cite{wenk2002phd} in 2002.
Concretely, we obtain the following two theorems:
\begin{restatable}{theorem}{mainresulttranslations}\label{thm:mainresulttranslations}   
    There exists an algorithm to compute the continuous Fréchet distance under translation between two time series of complexity $n$ in time $\Oh(n^{8/3} \log^3 n)$.
\end{restatable}
\begin{restatable}{theorem}{mainresultscaling}\label{thm:mainresultscalings}  
    There exists an algorithm to compute the continuous Fréchet distance under scaling between two time series of complexity $n$ in time $\Oh(n^{8/3} \log^3 n)$.
\end{restatable}
These results are made possible by the introduction of a novel framework for studying (continuous) time series~\cite{BD24} and the offline dynamic grid reachability data structure of~\cite{DBLP:journals/talg/BringmannKN21}.
Our approach essentially reduces (in an algorithmic sense) the continuous problems to their discrete counterparts and hence surprisingly matches the results for the two distance measures in the discrete setting.
We note that in the Euclidean plane, there still is a large gap between the discrete and continuous setting and, while there was significant process on the discrete side, there was no progress for the continuous Fréchet distance under translation or scaling.
We believe that our approach is also of independent interest and can potentially be applied to other continuous 1-dimensional Fréchet distance problems, to match the running times of the discrete setting.

\subparagraph*{The asymmetric case.} We remark that the offline dynamic grid reachability data structure of~\cite{DBLP:journals/talg/BringmannKN21} is only described for the symmetric case in which both curves have the same complexity. While we believe that this restriction can be eliminated, it would drastically complicate the exposition of our results, which is why we also make this assumption in the statements of our results.
Otherwise this restriction would not be necessary.


\ \\
We note that in an independent work, progress was made on the Fréchet distance under transformations in general dimensions~\cite{BBHNW25}. In this work, the authors present an $\tilde{\Oh}(n^{3k + 4/3})$ algorithm for the Fréchet distance under transformations with $k$ degrees of freedom, making use of the offline dynamic data structure of~\cite{DBLP:journals/talg/BringmannKN21}.

\subsection{Technical Overview}
On a high level, our approach for the decision problem for a given $\delta>0$ for both the translation and scaling invariant Fréchet distance can be summarized in the following three steps.
\begin{enumerate}
    \item The curves are represented by their (slightly adapted) $\delta$-signatures~\cite{DBLP:conf/soda/DriemelKS16}, which are simplifications at the given length-scale $\delta$ and encode the large scale behavior of the curve in terms of a subset of selected vertices that essentially form a discrete curve.
    Curves naturally decompose into three parts, the beginning, middle, and end, with the middle part given by the $\delta$-signature.
    The task is to design and choose data structures for all parts in such a way that they can be efficiently updated and combined when the transformation changes.
    \item To process the beginning and end of the curves, we introduce another simplification, and subsequently identify the first point in the beginning of one curve that matches the entire beginning of the other, and likewise for the end parts, where we instead identify the last point that allows a matching. The crux is that we need to be able to do this efficiently for different transformations. We achieve this by introducing the structural notion of deadlocks. 
    \item We identify and precompute a set of $\Oh(n^2)$ representative transformations at which the answer to the decision problem is subject to change, and for which the answer can be updated efficiently as we sweep over them. Since each part of the curves essentially behaves like a discrete curve, the decision problem ultimately reduces to a reachability question for an object very similar to the well-known free space matrix used in the decision problem for the discrete Fréchet distance. As the best algorithm for the discrete Fréchet distance under translation, we then make use of the offline dynamic grid reachability data structure~\cite{DBLP:journals/talg/BringmannKN21} to implement updates and queries efficiently.
\end{enumerate}

\subparagraph*{A note on parametric search.}
As is common, in this work we mostly describe how we solve the decision version, to then subsequently apply parametric search \cite{Cole87}.
We note that the algorithm for the discrete Fréchet distance under translation on time series of~\cite{Filtser2020} cleverly avoids parametric search, and hence shaves a logarithmic factor that would otherwise appear when directly applying~\cite{DBLP:journals/corr/AvrahamKS15}.
One reason why we cannot take a similar approach --- ignoring the fact that this technique is used for the discrete setting --- is that the avoidance of parametric search in~\cite{Filtser2020} leads to updates that crucially depend on the results of the decider calls; the updates are therefore \emph{online}.
Hence, a combination of~\cite{Filtser2020} with the \emph{offline dynamic} data structure of~\cite{DBLP:journals/talg/BringmannKN21} seems infeasible.


\section{Preliminaries}\label{sec:prelims}
We denote with $[n]$ the set $\{1, \dots, n\}$.
For any two points $p_1, p_2\in \RR$, we denote with $\overline{p_1p_2}$ the directed line segment from $p_1$ to $p_2$. 
A \emph{time series} of complexity $n$ is a 1-dimensional curve formed by $n$ points $P(1),\ldots,P(n)\in\RR$ and the ordered line segments $\overline{P(i)P(i+1)}$ for $i=1,\ldots,n-1$. 
Such a time series can be viewed as a function $P:[1, n]\rightarrow\RR$, where $P(i+\alpha)=(1-\alpha)P(i)+\alpha P(i+1)$ for $i=1, \dots, n-1$ and $\alpha\in [0,1]$. Sometimes, we denote $P$ as $\langle P(1), P(2), \dots, P(n)\rangle$. The points $P(i)$ are called \emph{vertices} of $P$ and $\overline{P(i)P(i+1)}$ are called \emph{edges} of $P$ for $i=1,\dots,n-1$. For $1\leq s\leq t\leq n$, we denote by $P[s,t]$ the subcurve of $P$ obtained from restricting the domain to the interval $[s,t]$. Further, we denote with $B(P, \delta)$ the set of points having distance at most $\delta$ to a point of $P$, i.e., $B(P, \delta)\coloneqq\{x\mid \exists a\in [1, n]\text{ s.t. }|x-P(a)|\leq \delta\}$. 
We define $\im(P)$ to be the image of the time series~$P$, meaning that $\im(P)\coloneqq\{P(i)\mid i\in[1,n]\}$. 

\subparagraph{Degenerate point configurations.}
Our results do not depend on the assumption that the time series are in general position, as all such occurrences can be easily dealt with using arbitrary tie-breaks, or \emph{symbolic perturbation}, where necessary. Every such instance is a result of degenerate vertex-vertex distances, and we discuss more concretely how they are dealt with when they arise.




To define the Fréchet distance, let $P$ and $Q$ be two time series of complexity $n$ and $m$. 
Further, let $\mathcal{H}_n$ be the set of all continuous non-decreasing functions $h:[0,1]\rightarrow[n]$ with $h(0)=1$ and $h(1)=n$. 
The \emph{continuous Fréchet distance} is defined as
\[d_F(P,Q)=\min_{h_P\in \mathcal{H}_n, h_Q\in \mathcal{H}_m}\max_{\alpha\in[0,1]}|P(h_P(\alpha))-Q(h_Q(\alpha))|.\]

In this paper, we aim to determine the optimal translation or scaling of a time series that minimizes the Fréchet distance.
For a time series $P=\langle P(1), \dots, P(n)\rangle$ and a value~$t\in\RR$, the \emph{translated time series $P_t$} is $\langle P(1)+t, \dots, P(n)+t\rangle$ and the \emph{\F under translation} is defined as
\[d_F^T(P,Q)=\min_{t\in\RR}d_F(P,Q_t).\]
For a value $s\in\RR$, we define the \emph{scaled time series $sP$} of $P$ to be $sP=\langle s P(1), \dots, sP(n)\rangle$ and the \emph{\F under scaling} is defined as 
\[d_F^S(P,Q)=\min_{s\in\RR_{\geq0}}d_F(P,sQ).\]

In contrast to the Fréchet distance under translation, the Fréchet distance under scaling is not symmetric. The above is the directed version, and one obtains the undirected version by minimizing over both directed variants. 
Whenever we just refer to the \Fm, we mean the continuous Fréchet distance.

\subsection{Signatures and Coupled Visiting Orders}
A $\delta$-signature of a time series $P$ is a time series that consists of important maxima and minima of $P$. The concept of signatures was introduced by Driemel, Krivo\v{s}ija and Sohler~\cite{DBLP:conf/soda/DriemelKS16}. Later, Blank and Driemel~\cite{BD24} extended the $\delta$-signature by adding one vertex in the beginning and one in the end to obtain an extended $\delta$-signature. In this paper, we only use extended $\delta$-signatures. See \Cref{fig:example_signature} for an example.
Hence, whenever we say signature we mean extended signature. 
To define extended $\delta$-signatures, we use the notion of $\delta$-monotone time series. Examples of $\delta$-monotone time series can be found in Figure~\ref{fig:monotonecurves}.

\begin{figure}
    \centering
    \includegraphics{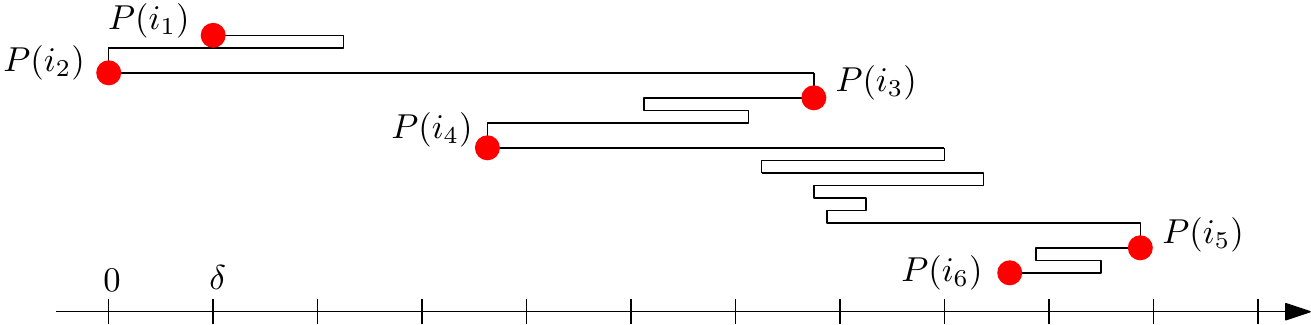}
    \caption{Throughout this paper, vertices of time series are drawn as vertical segments for clarity. The red vertices of the time series $P$ are its $\delta$-signature vertices. After linearly interpolating those red vertices, we get the extended $\delta$-signature of $P$.}
    \label{fig:example_signature}
\end{figure}

\begin{figure}
    \centering
    \begin{subfigure}[b]{0.48\textwidth}
        \centering
        \includegraphics[page=1]{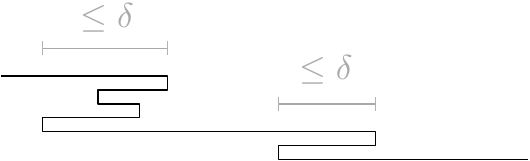}
    \end{subfigure}
    \centering
    \begin{subfigure}[b]{0.48\textwidth}
        \centering
        \includegraphics[page=2]{Figures/delta-monotone.pdf}
    \end{subfigure}
    \caption{The left (resp. right) time series is $\delta$-monotone increasing (resp. decreasing).} 
    \label{fig:monotonecurves}
\end{figure}
\begin{definition}
A time series $P:[1, n]\rightarrow \mathbb{R}$ is \emph{$\delta$-monotone increasing (resp. decreasing)} if it holds that $P(s)\leq P(s')+\delta$ (resp. $P(s)\geq P(s')-\delta$)  for all $s<s'\in [1, n]$.
\end{definition}
\begin{definition}[extended $\delta$-signature]
    Let $P=\langle P(1), \dotso, P(n)\rangle$ be a time series and $\delta\geq 0$.
    Then, an \emph{extended $\delta$-signature} $P'=\langle P(i_1), \dotso, P(i_t)\rangle$ with $1=i_1\leq i_2<\ldots<i_{t-1}\leq i_t=n$ of $P$ is a time series with the following properties:
    \begin{enumerate}[(a)]
        \item \emph{(non-degenerate)} For $k=2, \dotso, t-1$, it holds that $P(i_k)\notin \overline{P(i_{k-1}) P(i_{k+1})}$.
        \item \emph{($2\delta$-monotone)} For $k=1, \dotso, t-1$, $P[i_k, i_{k+1}]$ is $2\delta$-monotone increasing or decreasing.
        \item \emph{(minimum edge length)} If $t>4$, then for $k=2, \dotso, t-2$,  $|P(i_k)-P(i_{k+1})|>2\delta$.
        \item \emph{(range)} 
        \begin{itemize}
            \item For $k=2, \dotso, t-2$, it holds that $\im(P[i_{k}, i_{k+1}])= \overline{ P(i_{k}) P(i_{k+1})}$, and
            \item 
            $\im(P[1, i_2])\subset [P(i_2), P(i_2)+2\delta]$ or $\im(P[1, i_2])\subset [P(i_2)-2\delta, P(i_2)]$, and 
            \item 
            $\im(P[i_{t-1}, n])\subset [P(i_{t-1}),$ $ P(i_{t-1})+2\delta]$ or $\im(P[i_{t-1}, n])\subset [P(i_{t-1})-2\delta, P(i_{t-1})]$.
        \end{itemize}
        
        
    \end{enumerate}
    The vertices of $P'$ are called \emph{$\delta$-signature vertices} of $P$.
\end{definition}

Driemel, Krivo\v{s}ija and Sohler~\cite{DBLP:conf/soda/DriemelKS16} showed that there always exists a $\delta$-signature and 
that it can be computed in $\mathcal{O}(n)$ time\footnote{They proved this statement for $\delta$-signatures, but it can easily be shown for extended $\delta$-signatures as~well.}.  
Bringmann, Driemel, Nusser and Psarros~\cite{DBLP:conf/soda/BringmannDNP22} defined $\delta$-visiting orders and Blank and Driemel~\cite{BD24} used this concept to define coupled $\delta$-visiting orders (see \Cref{fig:coupled_visiting_order}), which can be used to decide whether the \F between two time series is at most $\delta$. 

\begin{definition}[coupled $\delta$-visiting order]\label{def:coupledVisitingOrder}
    Consider two time series $P=\langle P(1), \dotso, P(n)\rangle$ and $Q=\langle Q(1), \dotso, Q(m)\rangle$. 
    Then, a $\delta$-visiting order of $Q$ on $P$ for the $\delta$-signature vertices $Q(j_1), \dotso, Q(j_{s_Q})$ is a sequence of indices $x_1\leq\cdots\leq x_{s_Q}$ such that $|P(x_k)-Q(j_k)|\leq \delta$ for all $k$.  We similarly define a $\delta$-visiting order $y_1\leq\cdots\leq y_{s_P}$ of $P$ on $Q$ for the $\delta$-signature vertices $P(i_1), \dotso, P(i_{s_P})$ of $P$.
    These two $\delta$-visiting orders are said to be crossing-free if there exists no $k, l$ such that $i_k< x_l$ and $j_l<y_k$, or $x_k< i_l$ and $y_l<j_k$. In this case, the ordered sequence containing all tuples $(x_k, j_k)$ and $(i_l, y_l)$, where $k=1, \dotso, s_Q$ and $l=1, \dotso, s_P$, is called \emph{coupled $\delta$-visiting order}.
\end{definition}

\begin{figure}
    \centering
    \includegraphics{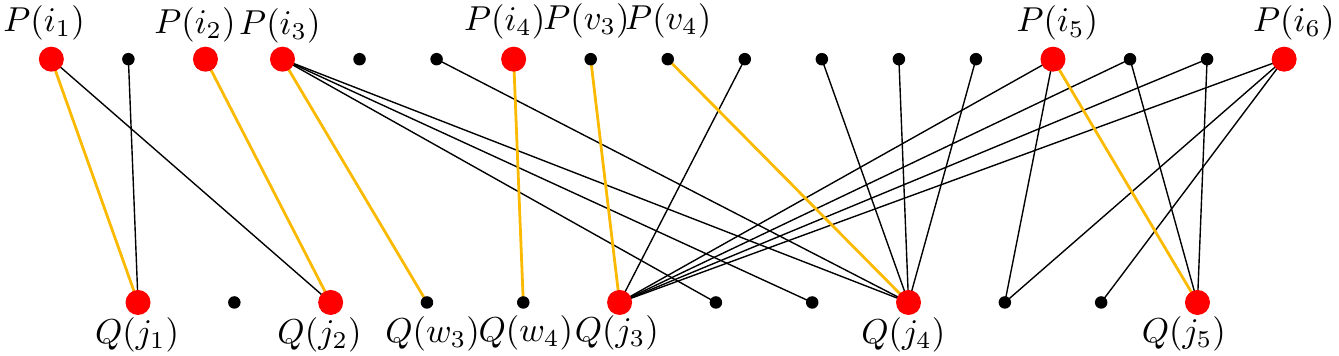}
    \caption{Visualization of a coupled $\delta$-visiting order. Edges are drawn between vertices $P(i)$ and $Q(j)$, when one is a $\delta$-signature vertex and $|P(i)-Q(j)|\leq \delta$. Then, a coupled $\delta$-visiting order consists of a subset of the drawn edges, shown as an orange bipartite graph, where no two edges cross and it contains one incident edge for every $\delta$-signature vertex.}
    \label{fig:coupled_visiting_order}
\end{figure}

\begin{restatable}[Lemma 9 of~\cite{BD24}]{lemma}{keylemma}\label{l:key_lemma}
    Let $P$, $Q$ be two time series and $\langle P(i_1), \dots, P(i_{s_P})\rangle$ and $\langle Q(j_1), \dots, Q(j_{s_Q})\rangle$ their $\delta$-signatures. Then, $d_F(P, Q)\leq \delta$ if and only if there exist a coupled $\delta$-visiting order $((v_1, w_1), (v_2, w_2), \dotso, (v_t, w_t))$ of $P$ and $Q$ such that 
    \begin{enumerate}[(i)]
        \item if $v_2=i_2$ (resp. $w_2=j_2$), 
        then $\exists w^*\in[w_2, \min\{w_2+1, j_2\}]$ (resp. $v^*\in[v_2, \min\{v_2+1, i_2\}]$) such that $d_F(P[1, v_2], Q[1, w^*])\leq \delta$ (resp. $d_F(P[1, v^*], Q[1, w_2])\leq \delta$), and
        \item if $v_{t-1}=i_{s_P-1}$ (resp. $w_{t-1}=j_{s_Q-1}$)
        , then $\exists w^{**}\in[\max\{j_{s_Q-1}, w_{t-1}-1\}, w_{t-1}]$ (resp. $v^{**}\in[\max\{i_{s_P-1}, v_{t-1}-1\}, v_{t-1}]$) such that $d_F(P[v_{t-1}, v_t], Q[w^{**}, w_t])\leq \delta$ (resp. $d_F(P[v^{**}, v_t], Q[w_{t-1}, w_{t}])\leq \delta$).
    \end{enumerate}
\end{restatable}

 \begin{proof}
    There is a small difference to the cited lemma. Consider the case that $v_2=i_2$. The case that $w_2=j_2$ is symmetric as well as for the second last tuple.
    The statement in this lemma requires that
    \begin{compactenum}
        \item[(1)] $Q(w_2)$ is a vertex and $\exists\ w^*\in[w_2, \min\{w_2+1, j_2\}]$ such that $d_F(P[1, v_2], Q[1, w^*])\leq \delta$.
    \end{compactenum}
    The requirement in the cited lemma is that
    \begin{compactenum}
        \item[(2)] $d_F(P[1, v_2], Q[1, w_2])\leq \delta$ and $Q(w_2)$ does not need to be a vertex of $Q$.
    \end{compactenum}
    We show that requiring $(1)$ or $(2)$ is equivalent. Let $w_2$ and $w^*$ be as in $(1)$. Then $w^*$ satisfies the condition of $(1)$. Now, let $w_2$ be as in $(2)$. Since $Q[1, j_2]$ is contained in some $\delta$-ball and $w_2\leq j_2$, it holds that $|P(i_2)-Q(\lfloor w_2\rfloor)|\leq \delta$ or $|P(i_2)-Q(\lceil w_2\rceil)|\leq \delta$. In the first case, $w_2=\lfloor w_2\rfloor$ and $w^*=w_2$ satisfy the conditions in $(1)$. In the second case, it holds that $\im(Q[w_2,\lceil w_2\rceil])\subset B(P(i_2), \delta)$. Hence, $w_2=\lceil w_2\rceil$ and $w^*=w_2$ satisfy the conditions in $(1)$.
\end{proof}

\subsection{Grid Reachability}
\begin{definition}[Offline Dynamic Grid Reachability]
Let $G$ be the directed $n\times n$-grid graph in which every node is either activated or deactivated. We are given updates $u_1,\ldots,u_U$, which are of the form \enquote{activate node $(i,j)$} or \enquote{deactivate node $(i,j)$} in an offline manner. The task of \emph{offline dynamic grid reachibility} is to compute for all $1\leq\ell\leq U$ if $(n,n)$ can be reached by $(1,1)$ after updates $u_1,\ldots,u_{\ell}$ are performed.
\end{definition}
Our main results depend on the following theorem.
\begin{theorem}[Theorem 3.1 of \cite{DBLP:journals/talg/BringmannKN21}] \label{thm:offline_dyn_grid_reachability}
Offline Dynamic Grid Reachability can be solved in time $\Oh(n^2 + U n^{2/3} \log^2 n)$.
\end{theorem}

\section{The Static Algorithm}
First, we discuss how the \emph{modified free-space matrix} can be used to decide whether the \F between two time series is at most a given threshold $\delta$. \Cref{l:key_lemma} implies that, except for the beginning and the end, it is enough to look at pairs $(P(i), Q(j))$ of vertices, where one of the two vertices is a $\delta$-signature vertex.

\begin{definition}[prefix and suffix]
    Let $P$ be a time series of complexity $n$ and let $P(i_2)$ denote the second $\delta$-signature vertex of~$P$, and $P(i_{s_P-1})$ denote the penultimate $\delta$-signature vertex of~$P$. Then $P[1,i_2]$ is the prefix $\pre{P}$ and $P[i_{s_P-1},n]$ in reverse order defines the suffix $\suf{P}$ of $P$. That is, the time series $P[i_{s_P-1},n]$ is $\widehat{\suf{P}}$, where $\widehat{P}$ denotes the time series defined by the vertices of $P$ in reverse order.
\end{definition}

Observe that by the definition of the $\delta$-signature vertices the time series $\pre{P}$ and $\suf{P}$ are each contained in some ball of radius $\delta$.

\begin{definition}[minimal matcher]
    Let $P$ and $Q$ be two time series of complexity $n$ and $m$ respectively. The \emph{\ePQ{P}{Q}} is the smallest $w\in[m]$ such that
    \begin{compactenum}
        \item $|P(n)-Q(w)|\leq\delta$, and
        \item there is a $w^*\in[w,\min(w+1,m)]$ such that $\mathrm{d}_F(P,Q[1,w^*])\leq\delta$. 
    \end{compactenum}
    If no such value exists, we set it to $\infty$.
\end{definition}
\begin{figure}
    \centering
    \includegraphics{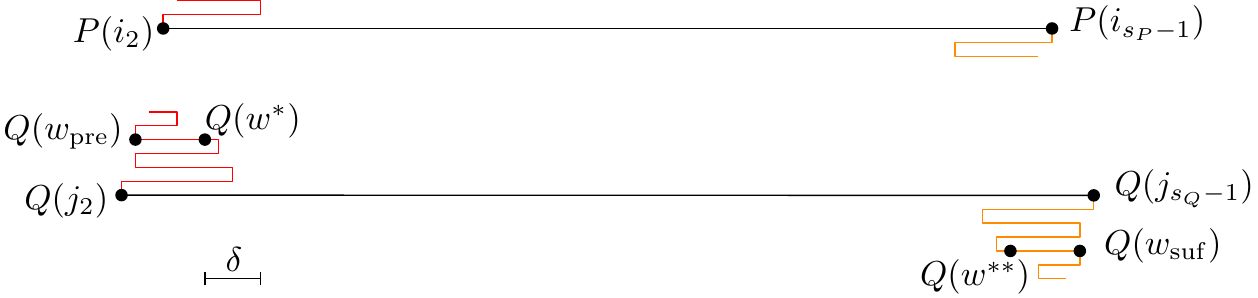}
    \caption{Minimal $\pre{P}$-matcher on $\pre{Q}$, $\wpre$, and \ePQ{\suf{P}}{\suf{Q}}, $\wsuf$.}
    \label{fig:w_pre}
\end{figure}

We denote by $\wpre$ the \ePQ{\pre{P}}{\pre{Q}} and by $\vpre$ the \ePQ{\pre{Q}}{\pre{P}}. For the suffix, we denote by $\wsuf$ the index of the vertex in $Q$ corresponding to the \ePQ{\suf{P}}{\suf{Q}} (and similarly $\vsuf$ when the roles of $P$ and $Q$ are swapped). See \Cref{fig:w_pre} for an example.

Due to the next observation, we will only discuss how to process the prefix in \Cref{sec:prefSuff}, \Cref{sec:PrefTranslation} and \Cref{sec:PrefScaling}. 
\begin{observation}\label{obs:ConnectionSufPre}
    Let 
    $\widehat{\wpre}$ be the \ePQ{\pre{\widehat{P}}}{\pre{\widehat{Q}}}. If $m$ is the complexity of $Q$, then $\wsuf=m-(\widehat{\wpre}-1)$.
\end{observation}

With these notions, we define the modified free-space matrix. See \Cref{fig:modified_free_space_matrix} for an example. 

\begin{definition}[Modified Free-Space Matrix]\label{def:mod_free_space_matrix} Let $P$ and $Q$ be two time series of complexity $n$ and $m$. Further, let $\langle P(i_1), \dotso, P(i_{s_P})\rangle$ and $\langle Q(j_1), \dotso, Q(j_{s_Q})\rangle$ be their $\delta$-signatures. 
We construct a matrix $M_{\delta}\in \{0, 1\}^{n\times m}$, where the entry $M_{\delta}(i,j)=1$ if
\begin{enumerate}[a)]
    \item  $P(i)$ and $Q(j)$ are both not $\delta$-signature vertices, or
    \item $|P(i)-Q(j)|\leq \delta$ and $(i, j)\notin [1, i_2]\times [1, j_2]\cup [i_{s_P-1}, n]\times [j_{s_Q-1}, m]$, or
    \item $|P(i)-Q(j)|\leq \delta$ and $(i, j)\in \{(1, 1), (n,m)\}$, or
    \item $|P(i)-Q(j)|\leq \delta$ and $(i, j)=(i_2, j_2)$ and $d_F(\pre{P}, \pre{Q})\leq \delta$, or
    \item $|P(i)-Q(j)|\leq \delta$ and $(i, j)=(i_{s_P-1}, j_{s_Q-1})$ and $d_F(\suf{P}, \suf{Q})\leq \delta$, or
    \item $(i, j)\in \{(i_2, \wpre), (\vpre, j_2), (i_{s_P-1}, \wsuf), (\vsuf, j_{s_Q-1})\}$. 
\end{enumerate}
Otherwise, the entry $M_{\delta}(i, j)=0$.
Further, we say that an entry $M_{\delta}(i, j)$ is reachable if there exists a traversal $(1, 1)=(a_1, b_1), (a_2, b_2), \dotso, (a_k, b_k)=(i, j)$ such that $M_{\delta}(a_l, b_l)=1$ and $(a_l, b_l)\in\{(a_{l-1}+1, b_{l-1}),(a_{l-1}, b_{l-1}+1),(a_{l-1}+1, b_{l-1}+1)\}$ for all $l=2, \dotso, k$.
\end{definition}

\begin{figure}
    \centering
    \includegraphics[page=1]{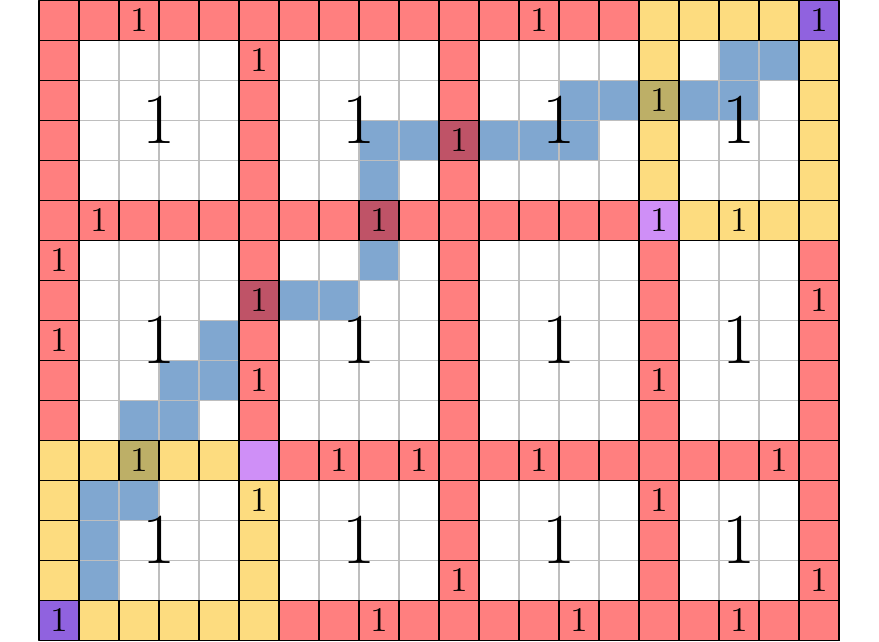}
    \caption{Example of a Modified Free-Space Matrix $M_{\delta}$. 
    The colored columns (resp. rows) correspond to the $\delta$-signature vertices of $P$ (resp. $Q$). The white entries are all $1$ by \Cref{def:mod_free_space_matrix}~a), the red entries are defined by~b), the purple entries by~c) and~d) and~e), and the yellow entries by~f). 
    The traversal drawn in blue uses only 1-entries of $M_\delta$. Hence, $M_\delta(n, m)$ is reachable.}
    \label{fig:modified_free_space_matrix}
\end{figure}
The importance of the modified free-space matrix is that it can be used to answer the decision problem of whether or not $d_F(P, Q)\leq \delta$ essentially in the same way as it is done for the discrete Fréchet distance using the free-space matrix. 
The next lemma follows mainly from a result of~\cite{BD24}. 

\begin{restatable}{lemma}{dfModMatEquiv}\label{lem:dfModMatEquiv}
    It holds that $d_F(P, Q)\leq \delta$ if and only if $M_{\delta}(n, m)$ is reachable. This can be checked in $\mathcal{O}(nm)$ time. 
\end{restatable}
\begin{proof}
    We use \Cref{l:key_lemma} to prove correctness. If $d_F(P, Q)\leq \delta$, then there exists a coupled $\delta$-visiting order $V=((v_1, w_1), \dots, (v_t, w_t))$ such that the requirements of \Cref{l:key_lemma} are fulfilled. If $Q(w_2)$ (resp. $P(v_2)$) is not the second $\delta$-signature vertex of $Q$ (resp. $P$), then $V$ is still a coupled $\delta$-visiting order with the properties of \Cref{l:key_lemma} after setting $w_2$ to $\wpre$ (resp. $v_2$ to $\vpre$). 
    Similarly, we redefine $w_{t-1}$ or $v_{t-1}$ depending on whether $P(w_{t-1})$ or $Q(v_{t-1})$ is not the second last $\delta$-signature vertex, maintaining a coupled $\delta$-visiting order with the properties of \Cref{l:key_lemma}. The resulting sequence defines a traversal from $M_{\delta}(1, 1)$ to $M_{\delta}(n,m)$ in the modified free-space matrix in the following way. For two consecutive tuples $(v_k, w_k)$, $(v_{k+1}, w_{k+1})$, we add \[ (v_k+1, w_k), (v_k+2, w_k), \dots, (v_{k+1}, w_k), (v_{k+1}, w_k+1), (v_{k+1}, w_k+2), \dots, (v_{k+1}, w_{k+1})\] in between them to the traversal. For all those new added tuples, it holds that their corresponding entry in the modified free-space matrix is $1$, because none of them corresponds to a $\delta$-signature vertex, by  definition of a coupled $\delta$-visiting order. In this way, we constructed a traversal from $M_{\delta}(1, 1)$ to $M_{\delta}(n, m)$.
    Hence, $M_{\delta}(n, m)$ is reachable. 

    Now, let $M_{\delta}(n, m)$ be reachable. Then, there exists a traversal $(1, 1)=(a_1, b_1), \dotso, $ $(a_t, b_t)=(n, m)$ such that $M_{\delta}(a_l, b_l)=1$ and $(a_l, b_l)\in\{{(a_{l-1}+1}, b_{l-1}), (a_{l-1}, {b_{l-1}+1}),$ $ ({a_{l-1}+1}, b_{l-1}+1)\}$ for all $l=2, \dotso, t$. Deleting all the tuples $(a_k, b_k)$ from this sequence, where neither $P(a_k)$ nor $Q(b_k)$ is a $\delta$-signature vertex, leads to a coupled $\delta$-visiting order, by definition. Further, it satisfies the properties of \Cref{l:key_lemma} by the definition of the modified free-space matrix.
    Therefore, by \Cref{l:key_lemma} it holds that $d_F(P, Q)\leq \delta$, which concludes the proof of correctness. 
    The running time follows by using dynamic programming and \Cref{lem:searchprefixend} and \Cref{lem:(i_2j_2)}.
\end{proof}

\subsection{Prefix and Suffix}\label{sec:prefSuff}
This section is dedicated to identifying the \ePQ{\pre{P}}{\pre{Q}}. We often use the fact that \pre{P} and $\pre{Q}$ are each contained in some $\delta$-ball. 
We will show that the vital behavior of such constricted time series can be captured by their extreme points.
\begin{definition}[extreme point sequence]
Let $P$ be a time series of complexity $n$, such that $P(n)$ is a global extremum of $P$ and $P$ is contained in some $\delta$-ball. We define the extreme point sequence $1=a_1< a_2 < \cdots< a_{\cpP}=n$ to be a 
sequence of indices, such that 
    \begin{itemize}
        \item (range-preserving) $\im(P[1, a_{k+1}])= \overline{P(a_k) P(a_{k+1})}$ for all $k$, and
        \item (extreme point) $P(a_k)=\min(P[1, a_k])$ or $P(a_k)=\max(P[1, a_k])$.
    \end{itemize}
\end{definition}

See \Cref{fig:extreme_point_sequence} and~\ref{fig:Prefix_Example} for an example of both an extreme point sequence and their associated preliminary assignments, defined next. 
\begin{figure}
    \centering
    \includegraphics[page=3]{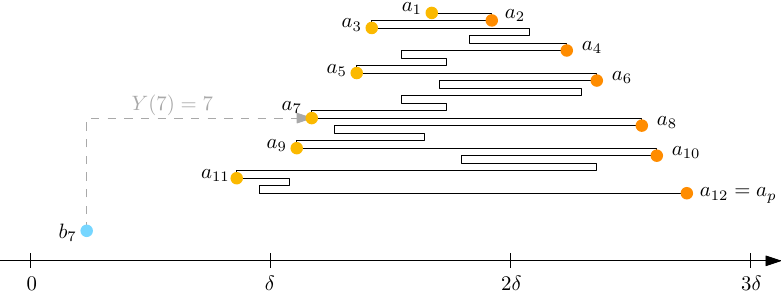}
    \caption{An extreme point sequence of the depicted time series $P$ is $1=a_1<a_2<\dots<a_{\cpP}$ and the preliminary assignment $\Y(7)$ of $Q$ on $P$ is $7$, since $Q(a_7)$ is the first vertex on $P$ in $B(Q(b_7), \delta)$.}
    \label{fig:extreme_point_sequence}
\end{figure}

\begin{definition}[preliminary assignment]
    Let $P$ and $Q$ be two time series such that the image of each time series is contained in some ball of radius $\delta$. Let $a_1<\ldots<a_{\cpP}$ and $b_1<\ldots<b_{\cpQ}$ be extreme point sequences of time series $P$ and $Q$ respectively. Define the \emph{preliminary assignment}  $\X(k)$ of $P$ on $Q$ for every $k\leq \cpP$ to be the smallest index such that $|P(a_k)-Q(b_{\X(k)})|\leq \delta$. If no such index exists, we set $\X(k)=\infty$. Let similarly $\Y(l)$ be the preliminary assignment of $Q$ on $P$. We say that $\X(k)$ and $\Y(l)$ form a \emph{deadlock}, if $l<{\X(k)}$ and $k<{\Y(l)}$. 
\end{definition}

Before we show how the preliminary assignment can be used to find the \ePQ{\pre{P}}{\pre{Q}}, we state some of its structural properties.

\begin{restatable}{lemma}{preliminaryAssignmentStructure}\label{lem:preliminaryAssignmentStructure}
    Let $P$ and $Q$ be two time series such that the image of each time series is contained in some ball of radius $\delta$.
    Let $a_1<\ldots<a_{\cpP}$ and $b_1<\ldots<b_{\cpQ}$ be extreme point sequences of time series $P$ and $Q$.  
    If $\X(1)=1$, then the following holds.
    \begin{enumerate}[i)]
        \item For any $3\leq k\leq p$ it holds that $\X(k)\geq \X(k-2)$.
        \item It holds that $\X(2k)=1$ for all $k$ or $\X(2k-1)=1$ for all $k$.
        \item If $k<p$ and $\X(k)\neq 1$, then for any index $l\geq \X(k)$ there is a $t\in [b_{l-1},b_l]$ such that $\im(P[1,a_{k+1}])\subset B(Q(t),\delta)$. 
        \item If $X(k)>1$ and $Y(l)$ form a deadlock, then $X(k)$ and $Y(X(k)-1)$ form a deadlock.
    \end{enumerate}
\end{restatable}
\begin{proof}
    Let $\X(k-2)>1$, since otherwise $i)$ trivially holds. This implies that $a_{k-2}\not\in B(b_1,\delta)$, which in turn implies the same for $a_k$, since $|a_k-b_1|\geq |a_{k-2}-b_1|$, by definition of the extreme point sequence and $a_1\in B(b_1,\delta)$. As $Q$ is range-preserving, $\X(k)\geq \X(k-2)$, which implies $i)$.

   For $ii)$, it suffices to note that $P[1, a_p]$ is contained in a ball of radius $\delta$ and $X(1)=1$.
    To prove $iii)$, let $k$ be the first index such that $\X(k)\neq 1$, that is, $a_k$ lies outside $B_\delta(b_1)$. Then $a_{k+1}$ must lie on the other side of $a_1$ than $a_k$, but at most a distance $2\delta$ away from $a_k$. As $\X(1)=1$, i.e., $a_1\in B(b_1,\delta)$, it holds that $a_{k+1}$ also lies inside $B_\delta(b_1)$, implying $\X(k+1)=1$. Now, $a_{k+2}$ must lie at least as far from $a_1$ as from $a_k$ and thus again lies outside $B_\delta(b_1)$, and so on for higher values of $k$.
    
    We thus have that $\max\{\X(k),\X(k+1)\}=\X(k)\leq l$. This in turn implies that \[\im(P[1,a_{k+1}])\subset \overline{P(a_k)\,P(a_{k+1})}\subset B(\overline{Q(b_{\X(k)})\,Q(b_{\X(k+1)})}, \delta)\subset B(\overline{Q(b_{l-1})\,Q(b_l)}, \delta).\]
    By definition, $P[1,a_{k+1}]$ is contained in a $2\delta$-range. Hence, there exists a value $t\in[b_{l-1},b_l]$ such that $\im(P[1,a_{k+1}])\subset B(Q(t), \delta)$, proving the claim.

    To prove $iv)$, let $\X(k)$ and $\Y(l)$ form a deadlock. Then, $l<\X(k)$ and $k<\Y(l)$. 
    By the symmetric argument of $i)$ for $\Y$, it holds that $k<\Y(l)\leq \Y(l+2)\leq \Y(l+4)\leq \cdots \Y(l')$ for $l'=\X(k)-1$ or $ l'= \X(k)$. By the definition of $X$ and $Y$, it holds that $Y(X(k))\leq k$. Hence, $k<\Y(\X(k)-1)$. As $\X(k)-1<\X(k)$, it holds that $X(k)$ and $Y(X(k)-1)$ form a deadlock.
\end{proof}

The importance of deadlocks is summarized in the following pivotal lemma.

\begin{restatable}{lemma}{startAndEnd}\label{l:startAndEnd}
    Let $P$ and $Q$ be time series such that the image of each time series is contained in some ball of radius $\delta$.
    Let $a_1<\ldots<a_{\cpP}$ and $b_1<\ldots<b_{\cpQ}$ be extreme point sequences of $P$ and $Q$. 
    For any $w^*\leq b_{\cpQ}$, it holds that ${d_F(P[1, a_{\cpP}], Q[1, w^*])\leq \delta}$ if and only if
    \begin{enumerate} [(i)]
        \item $|P(1)-Q(1)|\leq \delta$ and $|P(a_{\cpP})-Q(w^*)|\leq \delta$, 
        \item $\im(P[1, a_{\cpP}])\subset B(Q[1, w^*], \delta)$, 
        \item $\im(Q[1, w^*])\subset B(P[1, a_{\cpP}], \delta)$, and
        \item $\X(k)$ and $\Y(l)$ do not form a deadlock for all $k=1, \dotso, \cpP$ and $l=1, \dotso, \cpQ$. 
    \end{enumerate}   
\end{restatable}

Observe that if $\X(1)$ and $\Y(1)$ do not form a deadlock then already $|P(1)-Q(1)|\leq \delta$.

\begin{proof}
    We begin by showing that $d_F(P[1, a_{\cpP}], Q[1, w^*])\leq\delta$ implies Condition (i)-(iv). Conditions (i)-(iii) hold by definition.
    Next, observe that if $\X(k)$ and $\Y(l)$ form a deadlock, then $|P(a_k)-Q(t)|>\delta$ for all $t\leq b_l$. However, in this case any traversal of $P$ and $Q$ would have to match $Q(b_l)$ with some $P(s)$ with $s<a_k$, since the part of $Q$ before $Q(b_l)$ has to be traversed before reaching $P(a_k)$ as $d_F(P[1, a_{\cpP}], Q[1, w^*])\leq\delta$. As $\X(k)$ and $\Y(l)$ form a deadlock, also $|Q(b_l)-P(s)|>\delta$ for all $s\leq a_k$. Therefore, any traversal induces a cost greater than $\delta$. Thus $\X(k)$ and $\Y(l)$ cannot form a deadlock, implying Condition (iv).

    We now turn to proving that Conditions (i)-(iv) imply that $d_F(P[1, a_\cpP], Q[1, w^*])\leq\delta$. 
    The remainder of this proof consists of identifying points, via the preliminary assignments of $P$ and $Q$, that induce a traversal of $P$ and $Q$ with cost $\leq\delta$. 
    
    Observe that Conditions~(ii) and~(iii) imply that $\X(k)<\infty$ for all $k\leq t$ and $\Y(l)<\infty$ for all $b_l\leq w^*$.
    For now, we assume $\X\not\equiv 1$. Let $\mathcal{K}=\{1\leq k\leq \cpP-2\mid \X(k)\neq \X(k+2)\}\cup \{\cpP-1\leq k \leq \cpP \mid \X(k)\neq 1\}$ be the set of \textit{transition indices}, at which $\X(k)$ is about to change when ignoring all 1-entries. To construct a matching between $P[1, a_p]$ and $Q[1, w^*]$, we define values $t_k$ and $s_k$ for all transition indices $k\in \mathcal{K}$ with the following properties. See \Cref{fig:Prefix_Example} and \Cref{fig:beginning}.
    \begin{claim}\label{c:def_sk_and_tk}
    There exist values $\{s_k\}_{k\in\mathcal{K}}\in[1, a_\cpP]$ and $\{t_k\}_{k\in\mathcal{K}}\in[1, w^*]$ such that
        \begin{enumerate}[a)]
        \item If $k<\cpP-1$, then there is a $t_k\in[b_{\X(k)-1},b_{\X(k)}]$ such that $\im(P[1,a_{k+1}])\subset B(Q(t_k), \delta)$ and $t_k$ is minimal in the sense that $\im(Q[b_{\X(k)-1},t_k])\subset\overline{Q(b_{\X(k)-1})\,Q(t_k)}$.
        \item If $k\geq \cpP-1$, then there is a $t_k\in[b_{\X(k)-1},\min\{w^*, b_{\X(k)}\}]$ such that $\im(P[1,a_\cpP])\subset B(Q(t_k), \delta)$.
        \item If $k<\cpP-1$, then there is an $s_k\in[a_k,a_{k+1}]$ such that $\im(Q[t_k,t_l])\subset B(P(s_k), \delta)$, where $l$ is the transition index after $k$.
        \item If $k\geq \cpP-1$, then there is an $s_k\in [a_{\cpP-1}, a_{\cpP}]$ such that $\im(Q[t_k,w^*])\subset B(P(s_k), \delta)$ and $\im(P[s_k,a_\cpP])\subset B(Q(w^*), \delta)$.
    \end{enumerate}
    \end{claim}
    
    Utilizing $a)$-$d)$, we construct a traversal of $P$ and $Q$ as follows: Let $k$ be the smallest transition index. Match $P[1,s_k]$ to $Q(t_k)$ with $t_k=1$ via $a)$. Now let $k$ be the last handled transition index and let $l$ be the next transition index after $k$. First match $Q[t_k,t_l]$ to $P(s_k)$ via $c)$ (lightblue in \Cref{fig:beginning}) and then match $P[s_k,s_l]$ to $Q(t_l)$ via $a)$ (yellow in \Cref{fig:beginning}). Lastly, let $k$ be the last transition index. So far, the constructed matching matches $Q[1,t_k]$ to $P[1,s_k]$. We now match $Q(t_k,w^*)$ to $P(s_k)$ and lastly $P[s_k,a_\cpP]$ to $Q(w^*)$ via $d)$ (darkblue and orange in \Cref{fig:beginning}). 

    If instead $\X\equiv 1$, then by Condition $(iii)$ and the range-preservedness of $P$ there must be a point $s$ in $[a_{\cpP-1},a_{\cpP}]$, such that $\im(Q[1,w^*])\subset B(P(s),\delta)$. Choosing a maximal such $s$ also ensures that $\im(P[s,a_{\cpP}])\subset B(Q(w^*),\delta)$, since there are no deadlocks. However, then there is a traversal of $P$ and $Q$, by first matching $P[1,s]$ to $Q(1)$, then matching $Q[1,w^*]$ to $P(s)$, and lastly matching $P[s,a_{\cpP}]$ to $Q(w^*)$, which concludes the proof.
\end{proof}

\begin{proof}[Proof of Claim \ref{c:def_sk_and_tk}]
    \ 
    \begin{enumerate}[a)]
    \item The existence of $t_k$ in a) is an immediate consequence of \Cref{lem:preliminaryAssignmentStructure} iii). Further, by \Cref{lem:preliminaryAssignmentStructure} i) and since $k$ is a transition index, it holds that $t_k\leq b_{\X(k)}\leq b_{\X(k+2)-1}$. Since $\im(P[1, a_p])\subset B(Q[1, w^*], \delta)$, it follows that $b_{\X(k+2)-1}<w^*$.
    
    \item If $w^*\geq b_{\X(k)}$, the existence of $t_k$ in b) follows from \Cref{lem:preliminaryAssignmentStructure} iii). 
    So, assume that $w^*<b_{\X(k)}$. It holds by definition that $Q(b_{\X(k)-1})$ is either the maximum or the minimum of $Q[1, w^*]$. Further, no point on $Q[1, b_{\X(k)-1}]$ has distance at most $\delta$ to $P(a_k)$, by the definition of $\X(k)$. Hence, by Condition (ii) of \Cref{l:startAndEnd} and since $Q$ is contained in a $\delta$-ball, the other global maximum or minimum of $Q[1, w^*]$ lies on $Q[b_{\X(k)-1}, w^*]$. Therefore, there exists a point $t\in[b_{\X(k)-1}, w^*]$ such that $\im(Q[1,w^*])\subset\overline{Q(b_{\X(k)-1})\,Q(t)}$. Again by Condition (ii) of \Cref{l:startAndEnd}, $\im(P[1,a_\cpP])\subset B(Q[1,w^*],\delta)=B(\overline{Q(b_{\X(k)-1})\,Q(t)}, \delta)$. As $P[1,a_\cpP]$ is contained in a $2\delta$-range, there exists a value $t_k$ such that $b_{\X(k)-1}\leq t_k\leq t\leq w^*$ and $\im(P[1,a_\cpP])\subset B(Q(t_k), \delta)$.

    \item Let $l$ be the transition index following $k$, so that $\X(k+2)=\X(l)\neq \X(k)$. By the minimal choice of $t_l$ and the range-preservedness of $Q$, it suffices to show that there is a $s_k$ such that $Q[1,b_{\X(l)-1}]$ and $Q(t_l)$ lie in $B(P(s_k),\delta)$, as this already implies that $\im(Q[t_k,t_l])\subset \im(Q[1,t_l])\subset B(P(s_k),\delta)$.
    
    As $\X(k+2)=\X(l)$, neither $\Y(\X(l)-1)$ nor $\Y(\X(l)-2)$ can equal $k+2$. As $\Y$ and $\X$ do not form any deadlocks, in fact $k+1\geq \max(\Y(\X(l)-2),\Y(\X(l)-1))$. Hence, by \Cref{lem:preliminaryAssignmentStructure} (with the roles of $P$ and $Q$ reversed) there is a point $s_k\in[a_k,a_{k+1}]$, such that $Q[1,b_{\X(l)-1}]$ lies in  $B(P(s_k),\delta)$. As $P(s_k)\in \im(P[1,a_{k+1}])\subset \im(P[1,a_{l+1}])$, in particular, $P(s_k)\in B(Q(t_l),\delta)$, and as such $Q(t_l)\in B(P(s_k),\delta)$. Thus, $s_k$ is the sought-after point.
    
    \item By Condition~(iii) of \Cref{l:startAndEnd} and because $Q[1,w^*]$ is contained in a $2\delta$-range there is a point $s\in[a_{\cpP-1},a_\cpP]$ such that $\im(Q[t_k,w^*])\subset B(P(s), \delta)$. By choosing $s_k$ to be the maximal such  $s$ it follows that $s_k\geq s_l$ for any transition index $l$. 
    The second property follows from the maximality of $s_k$, as then $\im(P[s_k,a_\cpP])=\overline{P(s_k)\,P(a_\cpP)}$, and both $P(s_k)$ and $P(a_\cpP)$ have distance at most $\delta$ to $Q(w^*)$, by definition and Condition~(i), respectively.\qedhere
    \end{enumerate}
\end{proof}

\begin{figure}
    \centering
    \includegraphics[page=2]{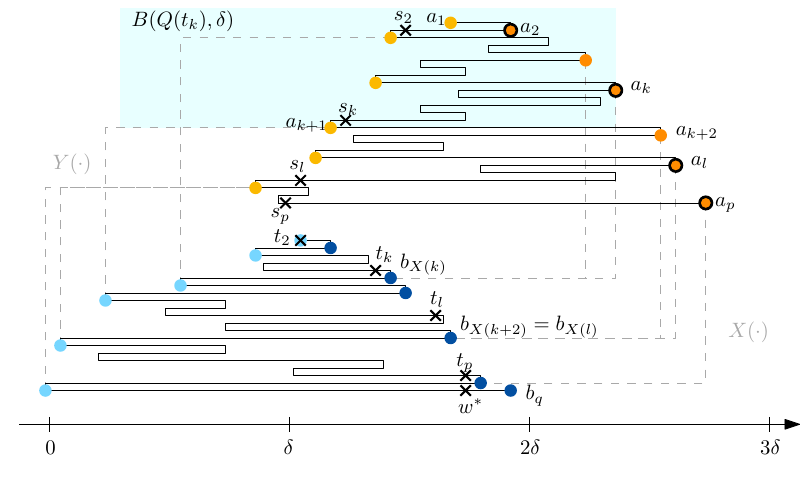}
    \caption{The gray dashed lines mark the preliminary assignments $X(\cdot)$ and $Y(\cdot)$. 
    The points defined in the proof of \cref{l:startAndEnd} are marked. The circled points are the transition indices.
    }
    \label{fig:Prefix_Example}
\end{figure}

\begin{figure}
    \centering
    \includegraphics{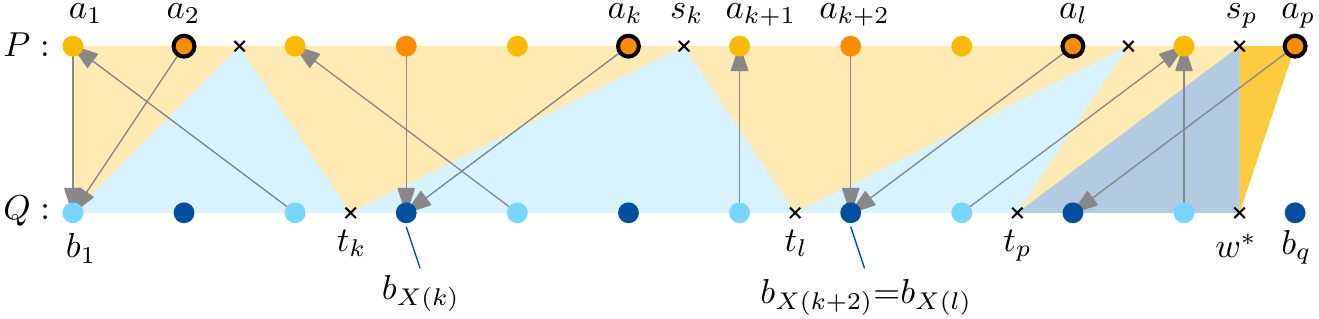}
    \caption{The transition indices are circled in black.
    The arrows link $a_k$ and $\X(k)$ as well as $b_k$ and $\Y(k)$ and for all $a_k$ (resp. $b_k$) without an outgoing arrow it holds $\X(k)=1$ (resp. $\Y(k)=1$). The traversal described in the proof of \cref{l:startAndEnd} is visualized by the triangles.}
    \label{fig:beginning}
\end{figure}
We use the next lemma to decide whether a \ePQ{\pre{P}}{\pre{Q}} exists.

\begin{restatable}{lemma}{wip}\label{lem:wip}

    Let $P$ and $Q$ be two time series and $P(i_2)$ be the second $\delta$-signature vertex~of~$P$ and $Q(j_2)$ of $Q$.
    There exists a \ePQ{\pre{P}}{\pre{Q}} if and only if $\X(k)< \infty$, and $\X(k)$ and $\Y(l)$ do not form a deadlock for all $k\in[\cpP]$ and $l\in[\cpQ]$.
    If it exists, it is 
    \[\min\{w=1, \dotso, j_2 \mid |P(i_2)-Q(w)|\leq \delta,\  \im(\pre{P})\subset B(Q([1, w+1], \delta)\}.\]
\end{restatable}
\begin{proof}
    If there exists $k$, $l$ such that $\X(k)$ and $\Y(l)$ form a deadlock, then by \cref{l:startAndEnd} there cannot exist such a value $w^*\in [1, j_2]$ such that $d_F(\pre{P}, Q[1, w^*])\leq \delta$. Hence, a \ePQ{\pre{P}}{\pre{Q}} does not exist.

    Now, assume that $\X(k)$ and $\Y(l)$ do not form a deadlock for all $k=1, \dotso, \cpP$ and $l=1, \dotso, \cpQ$.
    Let $\X_{\max}=\max\{\X(\cpP-1), \X(\cpP)\}$. If $X_{\max}=1$, then $\im(\pre{P})\subset B(P(1), \delta)$. Hence, the \ePQ{\pre{P}}{\pre{Q}} is $P(1)$. 
    Otherwise by \Cref{lem:preliminaryAssignmentStructure} (iii), there exists a value $t\in [b_{\X_{\max}-1}, b_{\X_{\max}}]$ such that $\im(\pre{P})\subset B(Q(t), \delta)$. Let $w$ be one of these values such that in addition
    $\im(Q[b_{\X_{\max}-1}, w])\subset \overline{Q(b_{\X_{\max}-1})Q(w)}$. Note that $w\leq j_2$. Since by assumption there does not exist a deadlock, $\max(\Y(\X_{\max}-2), \Y(\X_{\max}-1))\leq \cpP$. Therefore, $Q(w)$, $Q(b_{\X_{\max}-1})$, and $Q(b_{\X_{\max}-2})$ lie in $B(\pre{P}, \delta)$ and \[\im(Q[1, w])\subset \overline{Q(b_{\X_{\max}-2})Q(b_{\X_{\max}-1})}\cup\overline{Q(b_{\X_{\max}-1})Q(w)}\subset B(\pre{P}, \delta).\] 
    Hence, all conditions of \cref{l:startAndEnd} are satisfied for $w$ and thus $d_F(\pre{P}, Q[1, w])\leq \delta$.
    As a result, $w^*=\min\{w\in[1, j_2]\mid d_F(\pre{P}, Q[1, w])\leq \delta\}$ is well-defined. It holds that $\im(Q[1, w^*])\subset \im(Q[1, w])\subset B(\pre{P}, \delta)$. Hence, 
    by \cref{l:startAndEnd} it holds that \[w^*=\min\{w\in[1, j_2]\mid |P(i_2)-Q(w)|\leq \delta,\ \im(\pre{P})\subset B(Q[1, w], \delta)\}.\] If $|P(i_2)-Q(\lfloor w^*\rfloor)|\leq \delta$, we define $\wpre=\lfloor w^*\rfloor$. 
    Then, $Q(\wpre)$ is the \ePQ{\pre{P}}{\pre{Q}}, since $w^*$ was defined as the minimum. Otherwise, we set $\wpre=\lceil w^*\rceil$. Since $\pre{Q}$ is contained in a $2\delta$-range, $|P(i_2)-Q(w^*)|\leq \delta$, and $|P(i_2)-Q(\lfloor w^*\rfloor)|> \delta$, it holds that $|P(i_2)-Q(\lceil w^*\rceil)|\leq \delta$. 
    Therefore, we have that $\im(Q[1, \wpre])\subset \im(Q[1, w^*])\cup \overline{Q(w^*) Q(\wpre)}\subset B(\pre{P}, \delta)$. Hence, all conditions of \cref{l:startAndEnd} are fulfilled and thus $d_F(\pre{P}, Q[1, \wpre])\leq \delta$, so that $\wpre$ is the \ePQ{\pre{P}}{\pre{Q}}. 
    Further, by the definition of $w^*$, it holds that \[\wpre=\min\{w\in\{1, \dotso, j_2\}|\ |P(i_2)-Q(w)|\leq \delta,\  \im(\pre{P})\subset B(Q([1, w+1], \delta)\}.\qedhere\] 
\end{proof}

Finally, we are ready to show how the \ePQ{\pre{P}}{\pre{Q}} can be computed.

\begin{restatable}{lemma}{searchprefixend}\label{lem:searchprefixend}
    Let $P$ and $Q$ be time series and let $\im (Q[1, j])$ be given for every $j=1,\dots, q$. Further, assume the preliminary assignments $X$ and $Y$ of $\pre{P}$ and $\pre{Q}$ do not form a deadlock. Then we can compute the \ePQ{\pre{P}}{\pre{Q}} in $\mathcal{O}(\log n)$ time.
\end{restatable}
\begin{proof}
    Let $P(i_2)$ be the last vertex of $\pre{P}$ 
    and $Q(j_2)$ be the last vertex of $\pre{Q}$.  
    Using~\cref{lem:wip}, we can compute the \ePQ{\pre{P}}{\pre{Q}} as 
    \[\wpre=\min\{w\in\{1, \dotso, j_2\}|\ |P(i_2)-Q(w)|\leq \delta, \im(\pre{P})\subset B(Q([1, w+1], \delta)\}.\] 
    
    The second requirement is monotone, so that binary search can be used to identify the minimum $w\in \{1,\dotso, j_2\}$ that satisfies it. To further compute $\wpre$, observe that since $Q$ is range-preserving, the first edge of $Q$ that intersects $B(P(i_2),\delta)$ can also be found using binary search. Then, observe that among the edges intersecting $B(P(i_2),\delta)$, if the starting point of an edge is inside $B(P(i_2),\delta)$, then this is also true for proceeding edges since $Q$ is range-preserving, so that the first such starting point can again be found using binary search. 
\end{proof}

The following lemma shows how to decide for two time series $P$ and $Q$ whether the entry $M_{\delta}(i_2, j_2)$ of the modified free-space matrix is $0$ or $1$. 
\begin{restatable}{lemma}{itwojtwo}\label{lem:(i_2j_2)}
    Let $P$ and $Q$ be two time series and let extreme point sequences of $\pre{P}$ and $\pre{Q}$ be given.
    If the preliminary assignments of $\pre{P}$ and $\pre{Q}$ do not form a deadlock,
    then we can decide whether $d_F(\pre{P}, \pre{Q})\leq \delta$ in $\Oh(1)$ time. Otherwise, it holds that $d_F(\pre{P}, \pre{Q})> \delta$.
\end{restatable}
\begin{proof}
    We prove this lemma using~\Cref{l:startAndEnd}. Property (i) can be checked in constant time. By the definition of the extreme point sequence, it holds that $\im(\pre{P})= \overline{P(a_{\cpP-1}) P(a_\cpP)}$ and $\im(\pre{Q})= \overline{Q(b_{\cpQ-1})Q(b_\cpQ)}$. Therefore, we can check in constant time whether Properties~(ii) and~(iii) hold.
\end{proof}

The results of this section show how, using deadlocks, we can fill in the modified free-space matrix. In the following, we show that deadlocks allow efficient processing of updates of the transformations, which we crucially exploit for our algorithms.   

\section{1D Continuous Fréchet Distance Under Translation}\label{sec:translations}
We first describe the reduction of the 1D continuous Fréchet distance under translation to an offline dynamic grid graph reachability problem. Subsequently we describe how we solve this problem for two time series $P$ and $Q$ both of complexity $n$ in time $\Oh(n^{8/3} \log^3 n)$.

\subsection{\texorpdfstring{\boldmath The $\Oh(n^2)$ Events}{The O(n\textasciicircum 2) Events}}\label{sec:ArrangTrans}
In this section, we show how to compute a set of $\Oh(n^2)$ translation values such that the \F under translation between $P$ and $Q$ is at most $\delta$ if and only if there exists a translation value $t$ in this set such that the \F between $P$ and $Q_t$ is at most~$\delta$. We start by observing that it is sufficient to consider vertex-vertex events. 

\begin{restatable}{lemma}{pointpointtranslationevents}\label{lem:pointpointtranslationevents}
    Let $M(P,Q)=(m_{i,j})_{(i,j)\in[n]\times[n]}$ be the matrix where $m_{i,j}=\mathbbm{1}_{|P(i)-Q(j)|\leq\delta}$. Let $t,t'\in\RR$ be given. If $M(P,Q_t)=M(P,Q_{t'})$, then $d_F(P,Q_t)\leq\delta\iff d_F(P,Q_{t'})\leq\delta$.
\end{restatable}
\begin{proof}
    By \Cref{lem:dfModMatEquiv}, it suffices to show that if $M(P,Q_t)=M(P,Q_{t'})$ then the modified free-space matrices of $P$ and $Q_t$, and $P$ and $Q_{t'}$ agree. Observe that these modified free-space matrices do agree if the following conditions hold
    \begin{compactenum}
        \item $\forall(i,j)\in[n]\times[n]:M(P,Q_t)_{i,j}= M(P,Q_{t'}){i,j}$, \label{equivTrans:item1}
        \item the set of indices
        defining the $\delta$-signature vertices of $Q_t$ and $Q_{t'}$ agree, \label{equivTrans:item2}
        \item the \ePQ{\pre{P}}{\pre{Q_t}} is the \ePQ{\pre{P}}{\pre{Q_{t'}}}, \label{equivTrans:item3}
        \item the \ePQ{\pre{Q_t}}{\pre{P}} is the \ePQ{\pre{Q_{t'}}}{\pre{P}}, \label{equivTrans:item4}
        \item the \ePQ{\suf{P}}{\suf{Q_t}} is the \ePQ{\suf{P}}{\suf{Q_{t'}}}, and
        \item the \ePQ{\suf{Q_t}}{\suf{P}} is the \ePQ{\suf{Q_{t'}}}{\suf{P}}.\label{equivTrans:item6}
    \end{compactenum}
    Observe that \Cref{equivTrans:item1} holds as $M(P,Q_t)=M(P,Q_{t'})$. Further \Cref{equivTrans:item2} holds as by definition the extended $\delta$-signature is translation-invariant. \Cref{equivTrans:item3} holds as by \Cref{lem:wip} the \ePQ{\pre{P}}{\pre{Q_t}} only depends on $M(\pre{P},\pre{Q_t})$ and hence it agrees with the \ePQ{\pre{P}}{\pre{Q_{t'}}} as $M(\pre{P},\pre{Q_t})=M(\pre{P},\pre{Q_{t'}})$ as they are submatrices of $M(P,Q_t)$ and $M(P,Q_{t'})$ respectively. Similarly \Cref{equivTrans:item4} to \Cref{equivTrans:item6} hold implying the claim.
\end{proof}

\begin{observation}\label{obs:arrangement}
    Let $P$ and $Q$ be two time series of complexity $n$. There are at most $\Oh(n^2)$ many values $t$ such that there are values $i, j\in[n]$ such that $|P(i)-Q_t(j)|=\delta$. These values are precisely the interval boundaries of intervals of the form $\left[P(i)-Q(j)-\delta,P(i)-Q(j)+\delta\right]$.
\end{observation}

We briefly discuss how to deal with degenerate cases, that is, cases where it is not true that for every $t\in\RR$ exactly two values ${i\neq j\in[n]}$ exist such that $|P(i)-Q_t(j)|=\delta$, in \Cref{app:symbpert}.

\begin{restatable}[translation representatives]{corollary}{CompuTranslationArrangement}\label{cor:CompuTranslationArrangement}
    There exist a sorted set $\mathfrak{T}\subset\RR$ containing $\Oh(n^2)$ points and
    computable in $\Oh(n^2\log n)$ time with the following properties.
    \begin{itemize}
        \item It holds that $d_F^T(P,Q)\leq\delta$ if and only if $\exists t\in\mathfrak{T}$ such that $d_F(P,Q_t)\leq\delta$. 
        \item For two consecutive $t,t'$ in $\mathfrak{T}$, there exist only one pair of indices $k, l\in [n]$ such that $|P(k)-Q_t(l)|\leq \delta$ and $|P(k)-Q_{t'}(l)|>\delta$, or the other way round. Further, for the set of all pairs of consecutive $t, t'\in \mathfrak{T}$, those indices can be computed in $\Oh(n^2 \log n)$ time. 
    \end{itemize}
\end{restatable}
We call the set $\mathfrak{T}$ of \cref{cor:CompuTranslationArrangement} $\Oh(n^2)$ points the \emph{translation representatives}.

\begin{proof}
    \Cref{obs:arrangement} implies that there is an arrangement of $\RR$ such that only one inter-point distance namely that of $k$ and $l$ is such that the arrangement boundary is induced by $|P(k)-Q_t(l)|=\delta$. 
    As this arrangement and a point from each interval in the arrangement can be computed in $\Oh(n^2\log n)$ the claim follows by \Cref{lem:pointpointtranslationevents}.
\end{proof}

\subsection{Dealing with degenerate cases using symbolic perturbation}\label{app:symbpert}

If the instance under consideration is degenerate, then in the set of intervals of the form $\left[P(i)-Q(j)-\delta,P(i)-Q(j)+\delta\right]$ there are at least two whose interval boundaries coincide in at least one point, by \Cref{obs:arrangement}. Let $\{[a_1,b_1],\ldots\}$ be this set of closed intervals. We introduce symbolic perturbation on every interval boundary $a_i$ and $b_i$ resulting in $\hat{a}_i$ and $\hat{b_i}$ and a total order such that, in particular, any right boundary $\hat{b}_i$ of such an interval is \emph{strictly} larger than a left boundary $\hat{a}_i$ of another interval if and only if $b_i\geq a_i$. Then there are $\Oh(n^2)$ (symbolically perturbed) values in $\RR$ such that between two successive values the membership in the symbolically perturbed intervals differs by exactly one. Hence, these symbolically perturbed values are such that for any $\hat{t}\in \RR$ exactly two values $i\neq j\in[n]$ exist such that $|P(i)-Q_{\hat{t}}(j)|=\delta$. 
Furthermore, for any value $t\in\RR$ with some interval membership in intervals $\{[a_i,b_i],\ldots\}$ there is a symbolically perturbed version of $\hat{t}$ that has the same membership in perturbed versions of the intervals $\{[\hat{a}_i,\hat{b}_i],\ldots\}$.

\subsection{Prefix and Suffix under Translation}\label{sec:PrefTranslation} 
In \cref{sec:prefSuff}, we showed that the preliminary assignment can be used to determine the \ePQ{\pre{P}}{\pre{Q}}. In particular, the \ePQ{\pre{P}}{\pre{Q}} can be computed in $\mathcal{O}(\log n)$ time, if we know whether there does exists a deadlock or not. Therefore, we define the translated preliminary assignment and show how to keep track of whether a deadlock exists while translating the time series $Q$ with the values in $\mathfrak{T}$.

\begin{definition}[translated preliminary assignment]
    Let $P$ and $Q$ be time series and let $t\in\RR$. Define the \emph{translated preliminary assignment} $\X_t$ to be the preliminary assignment of $\pre{P}$ on $\pre{Q_t}$. Let similarly $\Y_t$ be the preliminary assignment of $\pre{Q_t}$ on $\pre{P}$. We say that $\X_t(k)$ and $\Y_t(l)$ form a \emph{deadlock}, if $l<{\X_t(k)}$ and $k<{\Y_t(l)}$. 
\end{definition}

Given two time series $P$ and $Q$ together with extreme point sequences $a_1,\ldots,a_p$ and $b_1,\ldots,b_q$ of $\pre{P}$ and $\pre{Q}$, we show that we can compute the set $\mathcal{T}\subset\mathfrak{T}$ of translations for which $X_t$ and $Y_t$ do not form a deadlock in total time $\Oh(n^2\log n)$. 

\begin{observation}
    For every $t\in \mathbb{R}$, the extreme point sequences of $Q_t$ and $Q$ coincide.
\end{observation}

\begin{observation}
    $\mathcal{T}\subset \{t\in\RR|X_t(1) = 1\} = \{t\in\RR|Y_t(1) = 1\}$.
\end{observation}

The next observation and next lemma show that the preliminary assignments of two consecutive translations of the translation representatives differ only slightly.

\begin{observation}\label{obs:oneStepTrans}
    Let two consecutive translation representatives $t,t'$ be given. Then, $X_t$ and $X_{t'}$ (resp. $Y_t$ and $Y_{t'}$) agree everywhere except at the index that participates in the point-point distance inducing the arrangement boundary separating $t$ and~$t'$. 
\end{observation}

\begin{restatable}{lemma}{updateStepTrans}\label{lem:updateStepTrans}
    Let two consecutive translation representatives $t,t'$ be given. 
    Let $k$ be the index of $P$ and $l$ the index of $Q$ participating in the point-point distance inducing the arrangement boundary separating $t$ and $t'$. 
    Then, it holds that
    \begin{compactenum}
        \item $|Y_t(l) - Y_{t'}(l)| \leq 2$,
        \item $Y_t(l) = \infty$ and $Y_{t'}(l) \in \{p-1,p\}$, or
        \item $Y_{t'}(l) = \infty$ and $Y_t(l) \in \{p-1,p\}$.
    \end{compactenum}
    Similarly, it holds that 
    \begin{compactenum}
        \item $|X_t(k) - X_{t'}(k)| \leq 2$,
        \item $X_t(k) = \infty$ and $X_{t'}(k) \in \{q-1,q\}$, or
        \item $X_{t'}(k) = \infty$ and $X_t(k) \in \{q-1,q\}$.
    \end{compactenum}
\end{restatable}
\begin{proof}
    We define the following nested family of closed sets \(\mathcal\{{F}_i=\bigcup_{i'\leq i}B(P(a_i),\delta)\mid 1\leq i\leq p\}\cup\{F_\infty=\RR\}\). 
    Let $\mathcal{A}$ be the arrangement of $\RR$ defined by these sets.
    By definition of $P$ and $a_1,\ldots,a_p$, each $F_i$ is a contiguous interval in $\RR$ and for $i\leq j-2$ it holds that $F_i$ and $F_j$ do not share a boundary. 
    Further, we have that 
    $Y_t(l) = \min\{k' \in [p]\cup\{\infty\}\mid  Q(b_l)+t\in F_{k'}\}$.  
    If $t,t'$ are two consecutive translation representatives, then $Q(b_l)+t$ and $Q(b_l)+t'$ must lie in neighboring intervals in $\mathcal{A}$ and thus their level in the set family can differ by at most two implying the claim for $Y_t$ and $Y_{t'}$. The claim for $X_t$ and $X_{t'}$ follows similarly.
\end{proof}

\begin{restatable}{lemma}{PrefixTranslation}\label{lem:PrefixTranslation}
    There is an algorithm that correctly computes the set $\mathcal{T}$ of translations $t\in \mathfrak{T}$ for which $X_t$ and $Y_t$ do not form a deadlock in $\Oh(n^2\log n)$. This is the set of translations in $\mathfrak{T}$ for which the preliminary assignments of $\pre{P}$ and $\pre{Q_t}$ do not form a deadlock. Similarly, we can compute the set of translations in $\mathfrak{T}$ for which $\suf{P}$ and $\suf{Q_t}$ do not form a deadlock in $\Oh(n^2\log n)$ time.
\end{restatable}
\begin{proof}
    We begin by computing the translation representatives $\mathfrak{T}$ of size $\Oh(n^2)$ and store it. This can be done in $\Oh(n^2\log n)$ time by~\cref{cor:CompuTranslationArrangement}. Next compute from the contiguous interval $[a,b]=\{t\in\RR|X_t(1) = 1\}$ in $\Oh(n^2)$ time the set $\mathcal{T}'=\mathfrak{T}\cap[a, b]$. 
    We define $\mathcal{D}_t=\{k\mid\exists l:\text{$X_t(k)$ and $Y_t(l)$ form a deadlock}\}$. 

    Initially, we compute $X_{t}$ and $Y_{t}$ for the smallest value $t\in \mathcal{T}'$. In addition, we store for every $1\leq k\leq \cpP$, whether $X_{t}(k)$ forms a deadlock with some $Y_{t}(l)$ or not, that is we decide whether $k\in\mathcal{D}_t$. Note that by \Cref{lem:preliminaryAssignmentStructure} (iv), it holds that $k\in \mathcal{D}_t$ if and only if $X_{t}(k)$ and $Y_{t}(X_{t}(k)-1)$ form a deadlock. Hence, this can be computed in $\mathcal{O}(n)$ time for all $k$. Further, we store the number $d_t = |\mathcal{D}_t|$. Then, $X_{t}$ and $Y_{t}$ do not form a deadlock if and only $d_t=0$.
    
    Now let the above data be given correctly for some translation $t\in\mathcal{T}'$ and let $t'$ be the next translation in $\mathcal{T}'$. Let $k\leq p$ and $l\leq q$ be the two indices that induced the arrangement boundary between $t$ and $t'$. We show how to update above data correctly.
    By \Cref{obs:oneStepTrans} and \Cref{lem:updateStepTrans}, for any $\hat{k}\in \mathcal{D}_t$ either $\hat{k} = k$ or $|\hat{k}-Y_t(l)|\leq 2$ or $\hat{k}\in \mathcal{D}_{t'}$. In the first two cases, we can check whether $\hat{k}\in\mathcal{D}_{t'}$ in $\Oh(1)$ time by \Cref{lem:preliminaryAssignmentStructure} (iv). Thus we can compute $d_{t'}$ given $d_t$ in time $\mathcal{O}(1)$.
    Repeating this update $\Oh(n^2)$ times sweeping over all translations in $\mathcal{T}'$ implies the claim. The claim for the suffix follows by~\Cref{obs:ConnectionSufPre}.
\end{proof}

\subsection{Solving the Decision Problem}
We are now ready to prove our result for the decision problem of the Fréchet distance under translation of time series.

\begin{lemma} \label{lem:fut_decider}
    We can decide the Fréchet distance under translation between two time series in time $\Oh(n^{8/3} \log^2 n)$.
\end{lemma}
\begin{proof}
The goal is to apply \cref{thm:offline_dyn_grid_reachability} to maintain reachability in the modified free-space matrix defined in \cref{def:mod_free_space_matrix} for all translations in $\mathfrak{T}$ from \cref{cor:CompuTranslationArrangement}.

In a pre-processing step, we first compute the signatures of both input time series and extreme point sequences of $\pre{P}, \pre{Q}, \suf{P}, \suf{Q}$, all in time $\Oh(n)$.
We use \cref{cor:CompuTranslationArrangement} to compute the set $\mathfrak{T}$ of candidate translations that we have to check reachability for, and fix the order in which we check them (e.g., from left to right).
Additionally, we invoke \cref{lem:PrefixTranslation} to compute the set $\mathcal{T} \subset \mathfrak{T}$ for which there are no deadlocks of the prefixes and suffixes.
The above steps take time $\Oh(n^2 \log n)$ each.
It now suffices to show that if we consider the the translations $\mathfrak{T}$ in order, there are only $\Oh(1)$ many updates per translation to the modified free-space matrix of \cref{def:mod_free_space_matrix}, for each type of entry (a) to (f). Furthermore, we need to determine all of these updates in time $\Oh(n^2 \log n)$.

Combining \cref{cor:CompuTranslationArrangement} and the information from the pre-computed signatures, we can pre-compute the $\Oh(1)$ many updates per translation in $\mathfrak{T}$ of type (a),(b), and (c) in time $\Oh(n^2 \log n)$.
For entries of type (d) and (e), we first check whether the translation at hand incurs a deadlock or not (i.e., we check for containment in the pre-computed set $\mathcal{T}$).
If it does not incur a deadlock, then we use \cref{lem:(i_2j_2)} to determine in constant time whether $d_F(\pre{P}, \pre{Q})\leq \delta$ and $d_F(\suf{P}, \suf{Q})\leq \delta$.
For type (f), we use the pre-computed signatures and \cref{lem:searchprefixend} to compute $\wpre, \vpre, \wsuf, \vsuf$ in $\Oh(\log n)$ time.

Note that by the above, each entry type incurs at most $\Oh(1)$ updates per translation, and we therefore have $\Oh(n^2)$ updates in total. All of these updates can be pre-determined, i.e., they are offline.
Consequently, we can apply~\Cref{thm:offline_dyn_grid_reachability} and obtain a running time of $\Oh(n^2 \cdot n^{2/3} \log^2 n) = \Oh(n^{8/3} \log^2 n)$.
\end{proof}

\subsection{Solving the Optimization Problem}\label{app:transoptimization}

Let us first consider a simple but inefficient variant to solve the optimization problem.
To that end, we define the set of all translations that align two vertices:
\[
T_{\text{align}} = \{ P(i) - Q(j) \mid i,j \in\{1, \dotso, n\} \}.
\]
By \cref{lem:pointpointtranslationevents}, for the optimal distance $\delta^*$ we have that there exist $t,t' \in T_{\text{align}}$ such that $|t-t'| = 2 \delta^*$.
Hence, for each pair $t, t' \in T_{\text{align}}$ we can test whether $d_F(P,Q + \frac{t+t'}{2}) \leq \frac{|t-t'|}{2}$, and the $t,t'$ with smallest $\frac{|t-t'|}{2}$ such that the decision query is positive is the Fréchet distance under translation $\delta^*$.
The above algorithm has running time $\Oh(|T_{\text{align}}|^2) = \Oh(n^4)$.
Fortunately, we can use parametric search to significantly reduce the running time and only incur an additional logarithmic factor compared to the running time of the decision problem.
We use a variant of Meggido's parametric search~\cite{Megiddo83} due to Cole~\cite{Cole87}, as was also used in~\cite{DBLP:journals/talg/BringmannKN21}.

For all $t_i \in T_{\text{align}}$, we define the functions $f_i^+(\delta) \coloneqq t_i + \delta$ and $f_i^-(\delta) \coloneqq t_i - \delta$.
These are the functions that we want to sort in our parametric search.
Note that, as argued above, the Fréchet distance under translation is realized for some $\delta^*$ such that there exist $i,j \in [|T_{\text{align}}|]$ with $f_i^+(\delta^*) = f_j^-(\delta^*)$.
To apply parametric search, we need two properties (see Section~3 in~\cite{Cole87}):
\begin{enumerate}
    \item Given any two functions $f_i^+(\delta), f_j^-(\delta)$, we can decide whether $f_i^+(\delta^*) \leq f_j^-(\delta^*)$ in time $T(n)$.\footnote{Note that we do not know the value of $\delta^*$!}
    \item Given a set of such comparisons, we can order them such that resolving any single comparison will resolve all comparisons before or after in the order. Also, determining the order of two \emph{comparisons} can be done in $\Oh(1)$ time.
\end{enumerate}

Consider the first property and assume $t_i < t_j$; all other cases are simple (as the comparison has the same result for all $\delta > 0$) or analogous.
Note that $f_i^+(\delta) = f_j^-(\delta)$ if and only if $\delta = \frac{t_j - t_i}{2}$, and $f_i^+(\delta) < f_j^-(\delta)$ for $\delta < \frac{t_j - t_i}{2}$ and $f_i^+(\delta) > f_j^-(\delta)$ for $\delta > \frac{t_j - t_i}{2}$.
Hence, using our decider of~\Cref{lem:fut_decider} to evaluate whether $\delta^* \leq \frac{t_j - t_i}{2}$, we can decide in which case we lie and hence decide $f_i^+(\delta^*) \leq f_j^-(\delta^*)$ in $T(n) = \Oh(n^{8/3} \log^2 n)$ time by~\Cref{lem:fut_decider} (without knowing $\delta^*$).

For the second property, consider that any pair of functions has a different threshold value at which the order flips.
Comparing the threshold value for two comparisons takes $\Oh(1)$ time.
If we now order the comparisons according to their threshold values, then deciding one of these comparisons will resolve either all comparisons before or after: if $f_1(\delta^*) \leq f_2(\delta^*)$, then this also holds for all comparisons with larger threshold value; if $f_1(\delta^*) > f_2(\delta^*)$, then this also holds for all comparisons with smaller threshold value.

Hence, we can apply parametric search as described in~\cite{Cole87} with a running time of $\Oh((|T_{\text{align}}| + T(n)) \cdot \log n) = \Oh((n^2 + n^{8/3}) \log^3 n) = \Oh(n^{8/3} \log^3 n)$, obtaining the following theorem:

\mainresulttranslations*

\section{1D Continuous Fréchet Distance Under Scaling}\label{app:scaling}
Our algorithm for the Fréchet distance under scaling works in a similar way as our algorithm for translations. An additional obstacle to overcome is that scaling a time series can change the $\delta$-signature.
We first describe the reduction of the 1D continuous \F under scaling to an offline dynamic grid graph reachability. Subsequently we describe how we solve this problem for two time series $P$ and $Q$ each of complexity $n$ in time $\Oh(n^{8/3}\log n)$.

\subsection{\texorpdfstring{\boldmath The $\Oh(n^2)$ Events}{The O(n\textasciicircum 2) Events}}

After translating a time series, the indices of the $\delta$-signature vertices remain the same. In contrast, the indices of the $\delta$-signature vertices can change when we scale a time series. Therefore, we begin by showing how the indices of the $\delta$-signature vertices change through scaling and show that there exists only $\Oh(n)$ different sets of indices of the $\delta$-signature vertices for scaled time series.

\begin{lemma}\label{lem:connectionSignature}
    Let $\langle Q(j_1), Q(j_2), \dots, Q(j_{s_Q})\rangle$ be the extended $\delta$-signature of $Q$. Then, the extended $s\delta$-signature of $sQ$ is $\langle sQ(j_1), sQ(j_2), \dots, sQ(j_{s_Q})\rangle$.
\end{lemma}
\begin{proof}
    This lemma follows directly by the definition of the extended $\delta$-signature, since $s|Q(j)-Q(j')|=|sQ(j)-sQ(j')|$ for all $j, j'\in[n]$. 
\end{proof}

\begin{restatable}{lemma}{lemorderSignature}\label{lemma:orderSignatures}
    We can compute in $\mathcal{O}(n\log n)$ time all values $s_j$ with $j=1, \dots, n$ such that the following holds. 
    The vertex $sQ(j)$ is a $\delta$-signature vertex of $sQ$ if and only if $s>s_j$.
\end{restatable}
\begin{proof}
    By Lemma 6.1 of \cite{DBLP:conf/soda/DriemelKS16}, there exists values $\delta_j$ with $j\in [n]$ such that $Q(j)$ is a $\delta$-signature vertex of $Q$ for every $\delta<\delta_j$ and no $\delta$-signature vertex for every $\delta\geq \delta_j$. Those values can be computed in $\Oh(n\log n)$ time by the proof of Theorem 6.1 of \cite{DBLP:conf/soda/DriemelKS16}. We set $s_j=\delta/\delta_j$ for every $j\in [n]$. By~\cref{lem:connectionSignature}, it holds that $sQ(j)$ is a $\delta=s\delta'$-signature vertex of $sQ(j)$ if and only if $Q(j)$ is a $\delta'$-signature vertex of $Q$. Hence, the claim follows.
\end{proof}

By~\cref{lemma:orderSignatures}, we get the following:
\begin{corollary}[coarse arrangement]\label{obs:coarseArragement}
    There exists an ordered set $\mathfrak{S'}$ of $\Oh(n)$ distinct intervals computable in $\Oh(n\log n)$ time such that the following holds.
    For any two $s,s'$ in the same interval of $\mathfrak{S}'$, the set of indices defining the $\delta$-signature vertices of $sQ$ and $s'Q$ agree. 
\end{corollary}

With the definition of the coarse arrangement, we can now proceed in a similar way to \Cref{sec:ArrangTrans}, to compute a set of scaling values that are sufficient to decide whether the \F under scaling is at most $\delta$.

\begin{lemma}\label{lem:pointpointscaleevents}
    Let $M(P,Q)=(m_{i,j})_{(i,j)\in[n]\times[n]}$ be the matrix where $m_{i,j}=\mathbbm{1}_{|P(i)-Q(j)|\leq\delta}$. Let $s,s'$ be in the same interval in $\mathfrak{S}'$. If $M(P,sQ)=M(P,s'Q)$, then ${d_F(P,sQ)\leq\delta}\iff d_F(P,s'Q)\leq\delta$.
\end{lemma}
\begin{proof}
    By \Cref{lem:dfModMatEquiv}, it suffices to show that if $M(P,sQ)=M(P,s'Q)$ and $M(sQ,sQ)=M(s'Q,s'Q)$ then the modified free-space matrices of $P$ and $sQ$, and $P$ and $s'Q$ agree. Observe that these modified free-space matrices do agree if the following conditions hold
    \begin{compactenum}
        \item $\forall(i,j)\in[n]\times[n]:M(P,sQ)_{i,j}= M(P,s'Q)_{i,j}$,\label{equivScale:item1}
        \item the set of indices 
        defining the $\delta$-signature vertices of $sQ$ and $s'Q$ agree,\label{equivScale:item2}
        \item the \ePQ{\pre{P}}{\pre{sQ}} is the \ePQ{\pre{P}}{\pre{s'Q}},\label{equivScale:item3}
        \item the \ePQ{\pre{sQ}}{\pre{P}} is the \ePQ{\pre{s'Q}}{\pre{P}},\label{equivScale:item4}
        \item the \ePQ{\suf{P}}{\suf{sQ}} is the \ePQ{\suf{P}}{\suf{s'Q}}, and
        \item the \ePQ{\suf{Q_t}}{\suf{P}} is the \ePQ{\suf{s'Q}}{\suf{P}}.\label{equivScale:item6}
    \end{compactenum}
    Observe that \Cref{equivScale:item1} holds as $M(P,sQ)=M(P,s'Q)$. Further, by definition of $\mathfrak{S}'$ and \Cref{obs:coarseArragement},
    \Cref{equivScale:item2} holds. 
    \Cref{equivScale:item3} holds as by \Cref{lem:wip} the \ePQ{\pre{P}}{\pre{sQ}} only depends on $M(\pre{P},\pre{sQ})$ and hence it agrees with \ePQ{\pre{P}}{\pre{s'Q}} as $M(\pre{P},\pre{sQ})=M(\pre{P},\pre{s'Q})$ as they are submatrices of $M(P,sQ)$ and $M(P,s'Q)$ respectively. Similarly \Cref{equivScale:item4} to \Cref{equivScale:item6} hold implying the claim.
\end{proof}

\begin{observation}\label{obs:arrangement1}
    Let $P$ and $Q$ be two time series. There are at most $\Oh(n^2)$ many values $s$ such that there are values $i\in[n]$ and $j\in[m]$ such that $|P(i)-sQ(j)|=\delta$. These values are precisely the interval boundaries of intervals of the form $\left[\frac{P(i)-\delta}{Q(j)},\frac{P(i)+\delta}{Q(j)}\right]$.
\end{observation}

We assume that for every $s\in\RR$ exactly two values $i\neq j\in[n]$ exist such that ${|P(i)-sQ(j)|=\delta}$. Otherwise the set of intervals $\left[\frac{P(i)-\delta}{Q(j)},\frac{P(i)+\delta}{Q(j)}\right]$ has two intervals that intersect in exactly one point. Let $\{[a_1,b_1],\ldots\}$ be this set of closed intervals. Introducing symbolic perturbation on every interval boundary $a_i$ and $b_i$ results, similarly to the case of translations, in a well-defined ordering of all the interval boundaries. These symbolically perturbed values are such that for any $\hat{s}$ exactly two values $i\neq j\in[n]$ exist such that $|P(i)-\hat{s}Q(j)|=\delta$.  Moreover, for any value $s\in\RR$ with some interval membership in intervals $\{[a_i,b_i],\ldots\}$ there is a symbolically perturbed version of $\hat{s}$ that has the same membership in perturbed versions of the intervals $\{[\hat{a}_i,\hat{b}_i],\ldots\}$.

\begin{corollary}[scaling representatives] \label{cor:CompuScalingArrangement}
    There exists a sorted set $\mathfrak{S}\subset \RR$ containing $\Oh(n^2)$ points and computable in $\Oh(n^2\log n)$ time with the following properties. 

    \begin{itemize}
    \item Every value $s_i$ of \Cref{lemma:orderSignatures} is contained in $\mathfrak{S}$.
        \item It holds that $d_F^S(P,Q)\leq\delta$ if and only if $\exists s\in\mathfrak{S}$ such that $d_F(P,sQ)\leq\delta$. 
    \item For two consecutive $s,s'$ in $\mathfrak{S}$, there exist only $\Oh(1)$ pairs of indices $k, l\in [n]$ such that $|P(k)-sQ(l)|\leq \delta$ and $|P(k)-s'Q(l)|> \delta$ or the other way round. Further, for the set of all pairs of consecutive $s, s'$ in $\mathfrak{S}$, these indices can be computed in $\Oh(n^2\log n)$ time.
    \end{itemize}
    
\end{corollary}
We call the set $\mathfrak{S}$ of \cref{cor:CompuScalingArrangement} the \emph{scaling representatives}.
\begin{proof}
    \Cref{obs:coarseArragement} and \Cref{obs:arrangement1} imply that there is an arrangement of $\RR$ such that for every neighboring cells of the arrangement such that the following holds. There exists indices $k$ and $l$ such that the arrangement boundary is induced by $|P(k)-sQ(l)|=\delta$ or $sQ(l)$ is not a $\delta$-signature vertex of $s Q$ and $s'Q(l)$ is a $\delta$-signature vertex of $s' Q$ for all $s'>s$.
    As this arrangement and a point from each interval in the arrangement can be computed in $\Oh(n^2\log n)$ the claim follows by \Cref{lem:pointpointscaleevents}.
\end{proof}

\subsection{Prefix and Suffix under Scaling}\label{sec:PrefScaling}
Similar to \Cref{sec:PrefTranslation}, we show how to keep track of whether a deadlock exists while scaling the time series $Q$ with the values in $\mathfrak{S}$.

\begin{definition}[scaled preliminary assignment]
    Let $P$ and $Q$ be two time series and let $s\in \RR$. 
    Define the \emph{scaled preliminary assignment} $X_s$ of $P$ and $Q$ to be the preliminary assignment of $\pre{P}$ and $\pre{sQ}$. 
    Let similarly $\Y_s$ be the preliminary assignment of $\pre{sQ}$ and $\pre{P}$. 
    We say that $\X_s(k)$ and $\Y_s(l)$ form a \emph{deadlock}, if $l<{\X_s(k)}$ and $k<{\Y_s(l)}$. 
\end{definition}

Given two time series $P$ and $Q$, we show that we can compute the set $\mathcal{S}\subset\mathfrak{S}$ of scaling factors for which the scaled preliminary assignments $X_s$ and $Y_s$ do not form a deadlock in total time $\Oh(n^2\log n)$.

\begin{observation}
    Let $s, s'$ lie in the same interval of $\mathfrak{S}'$. Then, the extreme point sequences of $\pre{sQ}$ is the same as of $\pre{s'Q}$.
\end{observation}

\begin{observation}
    $\mathcal{S}\subset \{s\in\RR\mid X_s(1) = 1\} = \{s\in\RR\mid Y_s(1) = 1\}$.
\end{observation}

The next observation and the next lemma show that the preliminary assignments of two consecutive scaling representatives that lie in the same interval of $\mathfrak{S'}$ differ only slightly.
\begin{observation}\label{obs:oneStepScal}
    Let two consecutive scalings representatives $s,s'$ 
    that lie in the same interval of $\mathfrak{S'}$ be given. Then, $X_s$ and $X_{s'}$ (resp. $Y_s$ and $Y_{s'}$) agree everywhere except at the index that participates in the point-point distance inducing the arrangement boundary separating~$s$ and~$s'$.
\end{observation}

\begin{lemma}\label{lem:updateStepScal}
    Let two consecutive scaling representatives $s,s'$ be given that lie in the same interval of $\mathfrak{S'}$. Let $a_1, \dots, a_p$ be an extreme point sequence of $\pre{P}$ and $b_1, \dots, b_q$ of $\pre{sQ}$.
    Let $k$ be the index of $P$ and $l$ be the index of $Q$ participating in the point-point distance inducing the arrangement boundary separating $s$ and $s'$. 
    Then, it holds that
    \begin{compactenum}
        \item $|Y_s(l) - Y_{s'}(l)| \leq 2$,
        \item $Y_s(l) = \infty$ and $Y_{s'}(l) \in \{p-1,p\}$, or
        \item $Y_{s'}(l) = \infty$ and $Y_s(l) \in \{p-1,p\}$.
    \end{compactenum}
    Similarly, it holds that
    \begin{compactenum}
        \item $|X_s(k) - X_{s'}(k)| \leq 2$,
        \item $X_s(k) = \infty$ and $X_{s'}(k) \in \{q-1,q\}$, or
        \item $X_{s'}(k) = \infty$ and $X_s(k) \in \{q-1,q\}$.
    \end{compactenum}
\end{lemma}
\begin{proof}
    We define the following nested family of closed sets \(\mathcal\{{F}_i=\bigcup_{i'\leq i}B(P(a_i),\delta)\mid 1\leq i\leq p\}\cup\{F_\infty=\RR\}\). 
    Let $\mathcal{A}$ be the arrangement of $\RR$ defined by these sets.
    By definition of $P$ and $a_1,\ldots,a_p$, each $F_i$ is a contiguous interval in $\RR$ and for $i\leq j-2$ it holds that $F_i$ and $F_j$ do not share a boundary. 
    Further, we have that 
    $Y_s(l) = \min\{k' \in [p]\cup\{\infty\}\mid  sQ(b_l)\in F_{k'}\}$.

    If $s$, $s'$ are two consecutive scaling representatives then $sQ(b_l)$ and $s'Q(b_l)$ must lie in neighboring intervals in $\mathcal{A}$ and thus their level in the set family can differ by at most two, implying the claim for $Y_s$ and $Y_{s'}$. The claim for $X_s$ and $X_{s'}$ follows similarly.
\end{proof}

\begin{lemma}\label{lem:PrefixScaling}
    There is an algorithm that correctly computes the subset $\mathcal{S}\subset\mathfrak{S}$ of scalings for which~$X_s$ and~$Y_s$ do not form a deadlock in $\Oh(n^2\log n)$. This is the set of scalings in~$\mathfrak{S}$ for which the preliminary assignments of~$\pre{P}$ and~$\pre{sQ}$ do not form a deadlock. Similarly, we can compute the set of scalings in $\mathfrak{S}$ for which $\suf{P}$ and $\suf{sQ}$ do not form a deadlock in $\Oh(n^2\log n)$ time.
\end{lemma}
\begin{proof}
    We begin by computing the scaling representatives $\mathfrak{S}$ and store it as $\Oh(n^2)$ sorted intervals in $\Oh(n^2\log n)$ time. Next, compute from the contiguous interval $[a,b]=\{s\in\RR\mid X_s(1) = 1\}$ the set $\mathcal{S}'=\mathfrak{S}\cap [a,b]$, in $\Oh(n^2)$ time. Then, $\mathcal{S}'$ consists of $\mathcal{O}(n^2)$ values $s_1<\ldots< s_z$. 

    Let $I\in\mathfrak{S}'$ (the coarse scaling arrangement). Then, for all $s, s'\in I\cap \mathcal{S}'$, the domain of $X_s$ and $Y_s$ does not change.
    We handle each of these interval $I$ separately. Let $k_I$ be the cardinality of $I\cap \mathcal{S}'$.
    We proceed in the same way as in the proof of~\cref{lem:PrefixTranslation}.
    First, define $\mathcal{D}_s=\{k\mid\exists l:\text{$X_s(k)$ and $Y_s(l)$ form a deadlock}\}$.
    Initially, we compute $X_{s}$ and $Y_{s}$ for the smallest value $s\in \mathcal{S}'\cap I$. In addition, we store for every $1\leq k\leq \cpP$, whether $X_{s}(k)$ forms a deadlock with some $Y_{s}(l)$ or not, that is we decide whether $k\in\mathcal{D}_s$. Note that by \Cref{lem:preliminaryAssignmentStructure} (iv), it holds that $k\in \mathcal{D}_s$ if and only if $X_{s}(k)$ and $Y_{s}(X_{s}(k)-1)$ form a deadlock. Hence, this can be computed in $\mathcal{O}(n)$ time for all~$k$. Further, we store the number $d_s = |\mathcal{D}_s|$. Then, $X_{s}$ and $Y_{s}$ do not form a deadlock if and only $d_s=0$.
    
    Now let the above data be given for some scaling factor $s\in\mathcal{S}'$ and let $s'$ be the next scaling in $\mathcal{S}'\cap I$. Let $k\leq p$ and $l\leq q$ be the two indices that induced the arrangement boundary between $s$ and $s'$. We show how to update above data correctly.
    By \Cref{obs:oneStepScal} and \Cref{lem:updateStepScal}, for any $\hat{k}\in \mathcal{D}_s$ either $\hat{k} = k$ or $|\hat{k}-Y_s(l)|\leq 2$ or $\hat{k}\in \mathcal{D}_{s'}$. In the first two cases, we can check whether $\hat{k}\in\mathcal{D}_{s'}$ in $\Oh(1)$ time by \Cref{lem:preliminaryAssignmentStructure} (iv). Thus we can compute $d_{s'}$ given $d_s$ in time $\mathcal{O}(1)$. 
    Hence, we compute in $\mathcal{O}(n\log n+k_I)$ time the scaling factors $s\in I\cap \mathcal{S}'$ for which $X_s$ and $Y_s$ do not form a deadlock. 
    It holds that $\sum_{I\in \mathcal{S}'}k_I\in \mathcal{O}(n^2)$ as $\mathfrak{S}'$ consists of disjoint intervals.
    As we do this for all intervals in $\mathfrak{S}'$, the total running time is in $\sum_{i\in \mathfrak{S'}}\mathcal{O}(n\log n+k_I)=\mathcal{O}(n^2\log n)$.
    The claim for the suffix follows by~\Cref{obs:ConnectionSufPre}.
\end{proof}

\subsection{Solving the Decision Problem}
We are now ready to prove our result for the decision problem of the Fréchet distance under scaling of time series.

\begin{lemma} \label{lem:fus_decider}
    We can decide the Fréchet distance under scaling between two time series in time $\Oh(n^{8/3} \log^2 n)$.
\end{lemma}
\begin{proof}
The goal is to apply \cref{thm:offline_dyn_grid_reachability} to maintain reachability in the modified free-space matrix defined in \cref{def:mod_free_space_matrix} for all scalings $\mathfrak{S}$ from \cref{cor:CompuScalingArrangement}. This set can be computed in $\Oh(n^2 \log n)$ time.
We scale the time series $Q$ with scaling factors in $\mathfrak{S}$ in increasing order. 

In contrast to the translation setting, for scalings, the signature of the scaled time series~$Q$ changes and hence also its prefix and suffix.
In a pre-processing step, we compute the $n$ different scalings where a new vertex is added to the signature of the scaled time series~$Q$ (which then stays in the signature until the end). This can be done in $\Oh(n \log n)$ time by \cref{lemma:orderSignatures}. Maintaining the extreme point sequence of $\pre{Q}, \suf{Q}$ over all scalings therefore also only takes $\Oh(n\log n)$ time.
We invoke \cref{lem:PrefixScaling} to compute the set $\mathcal{S} \subset \mathfrak{S}$ for which there are no deadlocks of the prefixes and suffixes using $\Oh(n^2 \log n)$ time.

The remaining proof is divided into two parts.
First, we argue that the changes to the $\delta$-signature (and therefore also to the prefix and suffix) that occur due to the scaling incur a quadratic number of updates in total to the modified free-space matrix.
Second, we show that all remaining updates (which are similar to the ones in the translation case) amount to $O(n^2)$ in total as well.

First, let us consider the updates that we have to perform on the modified free-space matrix of \cref{def:mod_free_space_matrix} for each type of entry (a) to (f) if the $\delta$-signature changes.
If, due to scaling, a new $\delta$-signature vertex is created, we might have to flip some entries that were previously of type (a); these amount to at most $\Oh(n)$ per new signature vertex, hence $\Oh(n^2)$ in total.
Also, if the vertex $j_2$ (resp. $j_{s_Q-1}$) becomes vertex $j_3$ (resp. $j_{s_Q-2}$) due to a new $\delta$-signature vertex, this potentially introduces some entries of type (b); these again amount to at most $O(n)$ entries per new signature vertex.
Since we recompute the entries of type (d) to (f) for each new scaling, we treat them in the next case.

For the second case, it remains to be shown that if we consider the scalings $\mathfrak{S}$ in order and we already applied the updates due to changes in the $\delta$-signature, then there are only $\Oh(1)$ many remaining updates per scaling to the modified free-space matrix of \cref{def:mod_free_space_matrix}, for each type of entry (a) to (f). Furthermore, we need to determine all of these updates in time $\Oh(n^2 \log n)$.
Combining \cref{cor:CompuScalingArrangement} and the information from the pre-computed signatures, we can pre-compute the $\Oh(1)$ many updates per scaling in $\mathfrak{S}$ of type (a),(b), and (c) in time $\Oh(n^2 \log n)$.
For entries of type (d) and (e), we first check whether the translation at hand incurs a deadlock or not (i.e., we check for containment in the pre-computed set $\mathcal{S}$).
If it does not incur a deadlock, then we use \cref{lem:(i_2j_2)} to determine in constant time whether $d_F(\pre{P}, \pre{Q})\leq \delta$ and $d_F(\suf{P}, \suf{Q})\leq \delta$.
For type (f), we use the pre-computed signatures and \cref{lem:searchprefixend} to compute $\wpre, \vpre, \wsuf, \vsuf$ in $\Oh(\log n)$ time.
Note that in this step we do not need to explicitly compute the scaled prefix and suffix.
By the above, each entry type incurs at most $\Oh(1)$ updates per translation, and we therefore have $\Oh(n^2)$ updates in total.

All of the updates above can be pre-determined, i.e., they are offline. Consequently, we can apply~\Cref{thm:offline_dyn_grid_reachability} and obtain a running time of $\Oh(n^2 \cdot n^{2/3} \log^2 n) = \Oh(n^{8/3} \log^2 n)$.
\end{proof}

\subsection{Solving the Optimization Problem}

Applying parametric search to the scaling variant is very similar to the translation variant. We nevertheless include a proof for completeness.
We again use the variant of Meggido's parametric search~\cite{Megiddo83} due to Cole~\cite{Cole87}, as was also used in~\cite{DBLP:journals/talg/BringmannKN21}.

%

For all $i,j \in [n]$, we define the functions $f_{i,j}^+(\delta) \coloneqq \frac{P(i) + \delta}{Q(j)}$ and $f_{i,j}^-(\delta) \coloneqq \frac{P(i) - \delta}{Q(j)}$.
These are the functions that we want to sort in our parametric search.
Let
\[
F = \{f_{i,j}^+(\delta) \mid i,j \in [n]\} \cup \{f_{i,j}^-(\delta) \mid i,j \in [n]\}
\]
be the set of all these functions.
Note that the Fréchet distance under scaling is realized for some $\delta^*$ such that there exist two distinct $f_1,f_2 \in F$ with $f_1(\delta^*) = f_2(\delta^*)$.
To apply parametric search, we need two properties (see Section~3 in~\cite{Cole87}):
\begin{enumerate}
    \item Given any two functions $f_1(\delta), f_2(\delta) \in F$, we can decide whether $f_1(\delta^*) \leq f_2(\delta^*)$ in time $T(n)$.
    \item Given a set of such comparisons, we can order them such that resolving any single comparison will resolve all comparisons before or after in the order. Also, determining the order of two \emph{comparisons} can be done in $\Oh(1)$ time.
\end{enumerate}

Consider the first property and consider $f_{i,j}^+(\delta), f_{k,\ell}^-(\delta)$ with $\frac{P(i)}{Q(j)} < \frac{P(k)}{Q(\ell)}$; all other cases are simple (as the comparison has the same result for all $\delta > 0$) or analogous.
By simple calculations, we obtain that $f_{i,j}^+(\delta) = f_{k,\ell}^-(\delta)$ if and only if
\[
\delta = \delta_{i,j,k,\ell} \coloneqq \frac{P(k)Q(j) - P(i)Q(\ell)}{Q(j) + Q(\ell)},
\]
and $f_{i,j}^+(\delta) < f_{k,\ell}^-(\delta)$ for $\delta < \delta_{i,j,k,\ell}$ and $f_{i,j}^+(\delta) > f_{k,\ell}^-(\delta)$ for $\delta > \delta_{i,j,k,\ell}$.
Hence, using our decider of~\Cref{lem:fus_decider} to evaluate whether $\delta^* \leq \delta_{i,j,k,\ell}$, we can decide in which case we lie and hence decide $f_i^+(\delta^*) \leq f_j^-(\delta^*)$ in $T(n) = \Oh(n^{8/3} \log^2 n)$ time by~\Cref{lem:fus_decider} (without knowing $\delta^*$).

For the second property, consider that any pair of functions has a different threshold value at which the order flips.
Comparing the threshold value for two comparisons takes $\Oh(1)$ time.
If we now order the comparisons according to their threshold values, then deciding one of these comparisons will resolve either all comparisons before or after: if $f_{i,j}^+(\delta^*) \leq f_{k,\ell}^-(\delta^*)$, then this also holds for all comparisons with larger threshold value; if $f_{i,j}^+(\delta^*) > f_{k,\ell}^-(\delta^*)$, then this also holds for all comparisons with smaller threshold value.

Hence, we can apply parametric search as described in~\cite{Cole87} with a running time of $\Oh((|F| + T(n)) \cdot \log n) = \Oh((n^2 + n^{8/3}) \log^3 n) = \Oh(n^{8/3} \log^3 n)$, obtaining the following theorem:

\mainresultscaling*

\bibliography{main}

\end{document}